\documentclass[10pt,pra,superscriptaddress,nofootinbib,twocolumn]{revtex4-2}
\pdfoutput=1

\usepackage[utf8]{inputenc}
\usepackage[english]{babel}
\usepackage[T1]{fontenc}
\usepackage{CJKutf8}
\usepackage{multirow}
\usepackage{subcaption}
\usepackage{amsmath,amssymb,amsthm,mathrsfs}
\usepackage{times,ragged2e}
\usepackage[dvipsnames]{xcolor}

\usepackage{tikz}
\usetikzlibrary{arrows.meta}
\newtheorem{theorem}{Theorem}[section]

\usepackage{braket}
\usepackage{graphicx}
\usepackage{mathtools}
\usepackage{bbold}
\usepackage{enumitem}
\usepackage[hyperindex,breaklinks,colorlinks=true,citecolor=blue,urlcolor=blue]{hyperref}
\usepackage{hhline}

\newcommand{\vecP}{\vec{P}}

\renewcommand{\L}{\mathcal{L}}

\usepackage[capitalize]{cleveref}

\begin{document}

\begin{CJK*}{UTF8}{bsmi}

\title{Hardy-type paradoxes for an arbitrary symmetric bipartite Bell scenario}

\author{Kai-Siang Chen}
\affiliation{Department of Physics and Center for Quantum Frontiers of Research \& Technology (QFort), National Cheng Kung University, Tainan 701, Taiwan}

\author{Shiladitya Mal}
\affiliation{Department of Physics and Center for Quantum Frontiers of Research \& Technology (QFort), National Cheng Kung University, Tainan 701, Taiwan}
\affiliation{Physics Division, National Center for Theoretical Sciences, Taipei 10617, Taiwan}
\affiliation{Centre for Quantum Science and Technology, Chennai Institute of Technology, Chennai 600069, India}

\author{Gelo Noel M. Tabia}
\affiliation{Department of Physics and Center for Quantum Frontiers of Research \& Technology (QFort), National Cheng Kung University, Tainan 701, Taiwan}
\affiliation{Physics Division, National Center for Theoretical Sciences, Taipei 10617, Taiwan}

\author{Yeong-Cherng Liang}
\email{ycliang@mail.ncku.edu.tw}
\affiliation{Department of Physics and Center for Quantum Frontiers of Research \& Technology (QFort), National Cheng Kung University, Tainan 701, Taiwan}
\affiliation{Physics Division, National Center for Theoretical Sciences, Taipei 10617, Taiwan}

\begin{abstract}
As with a Bell inequality, Hardy's paradox manifests a contradiction between the prediction given by quantum theory and local-hidden variable theories. 
In this work, we give two generalizations of Hardy's arguments for manifesting such a paradox to an {\em arbitrary}, but symmetric Bell scenario involving two observers. Our constructions recover that of Meng {\em et al.} [Phys. Rev. A. {\bf 98}, 062103 (2018)] and that first discussed by Cabello [Phys. Rev. A {\bf 65}, 032108 (2002)] as special cases. Among the two constructions, one can be naturally interpreted as a demonstration of the failure of the transitivity of implications (FTI). Moreover, a special case of which is equivalent to a ladder-proof-type argument for Hardy's paradox. Through a suitably generalized notion of success probability called degree of success, we provide evidence showing that the FTI-based formulation exhibits a higher degree of success compared with all other existing proposals. Moreover, this advantage seems to persist even if we allow imperfections in realizing the zero-probability constraints in such paradoxes. Explicit quantum strategies realizing several of these proofs of nonlocality without inequalities are provided. 
\end{abstract}
\date{\today}
\maketitle

\section{Introduction}

In the thought-provoking paper by Einstein, Podolksy, and Rosen~\cite{EPR35}, the strong correlations between measurement outcomes have led them to suspect that quantum theory could be somehow completed (with additional variables). This was eventually shown to be untenable by Bell~\cite{Bell64}, who proved that {\em no} local-hidden-variable (LHV) models can reproduce all quantum-mechanical predictions. In particular, he demonstrated how, with the help of so-called Bell inequalities, one can experimentally falsify the predictions of LHV models. Nowadays, we know that Bell nonlocality not only opens the door to answer fundamental questions in physics but also serves as an important resource for device-independent quantum information~\cite{Brunner_RevModPhys_2014,Scarani_DIQI_12}.

Interestingly, Bell inequalities are not the only way to manifest Bell nonlocality~\cite{Brunner_RevModPhys_2014}. Indeed, Greenberger, Horne, and Zeilinger (GHZ)~\cite{Greenberger:1989aa} showed in their seminal work that a logical contradiction can be demonstrated between the quantum mechanical prediction on a four-qubit GHZ state and that of {\em any} deterministic LHV model (DLHVM). Soon after, such a contradiction was also provided for a three-qubit GHZ state~\cite{GHSZ} and a two-qubit singlet state~\cite{Peres:1990ab}. This last construction, in particular, was adapted to give the well-known Peres-Mermin game~\cite{Mermin_90} for showing quantum pseudo-telepathy.

A common feature of these logical proofs is that they rely strongly on the perfect correlation of {\em maximally entangled} states. In contrast, Hardy~\cite{Hardy92} provided a different type of logical proof of ``nonlocality without inequality'' for a {\em partially entangled} two-qubit state. In Hardy's proof, a contradiction comes about only when a certain event is observed, see \cref{LogicalStructure_Hardy_Stapp}. The probability at which this event occurs is thus commonly called the {\em success probability}, as it facilitates (initiates) the chain of logical reasoning in Hardy's arguments.
 
Hardy's original proof was soon generalized to cater for certain bipartite quantum states of {\em arbitrary} local Hilbert space dimension~\cite{Clifton92}, an arbitrary number of qubit systems~\cite{Pagonis92} (see also~\cite{Jiang:PRL:2018,Luo:2018aa}), an arbitrary {\em partially entangled} two-qubit state~\cite{Hardy93}, and later to an experimental scenario involving an arbitrary number of binary-outcome measurement settings~\cite{Boschi97}. In the meantime, Stapp's reformulation~\cite{Stapp93} of Hardy's argument (which leads to the so-called Hardy paradox) made clear~\cite{Liang:PRep} that the paradox can also be interpreted as the failure of the transitivity implications (FTI), thereby demonstrating Bell nonlocality.

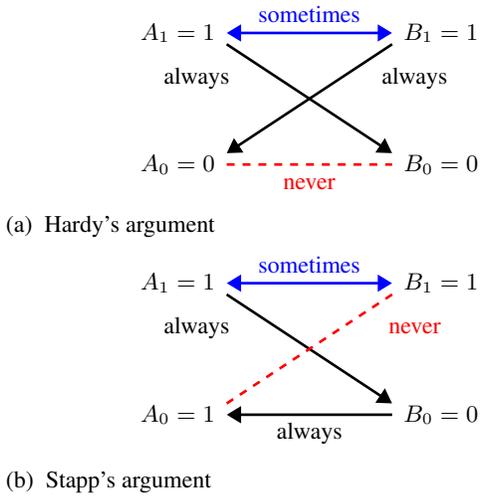
\begin{figure}[h!tbp]
\captionsetup{justification=RaggedRight,singlelinecheck=off}
\begin{subfigure}{0.45\textwidth}
\centering
\begin{tikzpicture}
    \node[] at (0,1.75) {$A_1=1$};
    \node[] at (3.5,1.75) {$B_1=1$};
    \node[] at (0,0) {$A_0=0$};
    \node[] at (3.5,0) {$B_0=0$};
    \node[blue] at (1.75,2) {sometimes};
    \node[] at (3.15,1.15) {always};
    \node[] at (0.25,1.15) {always};
    \node[red] at (1.75,-0.25) {never};
    \draw[Triangle-Triangle,blue,line width=1pt] (0.65,1.75)--(2.85,1.75);
    \draw[-Triangle,line width=1pt] (0.65,1.6)--(2.85,0.15);
    \draw[-Triangle,line width=1pt] (2.85,1.6)--(0.65,0.15);
    \draw[dashed,red,line width=1pt] (0.65,0)--(2.85,0);
\end{tikzpicture}
\caption{\label{Fig:Hardy} Hardy's argument}
\end{subfigure}
\begin{subfigure}{0.45\textwidth}
\centering
\begin{tikzpicture}
    \node[] at (0,1.75) {$A_1=1$};
    \node[] at (3.5,1.75) {$B_1=1$};
    \node[] at (0,0) {$A_0=1$};
    \node[] at (3.5,0) {$B_0=0$};
    \node[blue] at (1.75,2) {sometimes};
    \node[] at (1.75,-0.25) {always};
    \node[red] at (3.15,1.15) {never};
    \node[] at (0.25,1.15) {always};
    \draw[Triangle-Triangle,blue,line width=1pt] (0.65,1.75)--(2.85,1.75);
    \draw[-Triangle,line width=1pt] (0.65,1.6)--(2.85,0.15);
    \draw[dashed,red,line width=1pt] (2.85,1.6)--(0.65,0.15);
    \draw[Triangle-,line width=1pt] (0.65,0)--(2.85,0);
\end{tikzpicture}
\caption{\label{Fig:Stapp} Stapp's argument}
\end{subfigure}
\centering
\caption{Schematic representation of (a) Hardy's argument and (b) Stapp's reformulation~\cite{Stapp93} for demonstrating the inadequacy of DLHVMs in reproducing quantum predictions. Here, $A_x$ and $B_y$ represent, respectively, the outcome ($0$ or $1$) observed by Alice and Bob when she performs the $x$-the measurement and he performs the $y$-th measurement. In a DLHVM, the implications (black arrows) and the forbidden event (red) imply that the event $A_1=B_1=1$ is also forbidden, yet quantum theory defies this implication.
}
\label{LogicalStructure_Hardy_Stapp}
\end{figure}

Several years later, relaxations of Hardy's original formulation were also proposed. For example, motivated by Kar's observation~\cite{Kar97} that {\em no} mixed two-qubit entangled states exhibit Hardy's paradox, Liang and Li~\cite{Liang03,Liang:2005vl} generalized Hardy's argument by relaxing one of the equality constraints to an inequality constraint (see also\cite{Cabello02}). Indeed, their construction allowed them to demonstrate a Hardy-type logical contradiction for certain mixed two-qubit states via a generalized notion of success probability, called degree of success in~\cite{Rai:PRA:2021}. Subsequently, Kunkri {\em et al.}~\cite{KCA+06} showed that this generalization could give a higher degree of success compared to the original formulation in~\cite{Hardy93}. A brief discussion of a further generalization from Cabello and that of Liang-Li to a scenario with an arbitrary number of measurement settings was subsequently given in~\cite{Cereceda16}.

Indeed, a noticeably higher success probability (or degree of success) can be obtained if we are willing to consider a Bell scenario with more measurement settings~\cite{Boschi97,Cereceda16}, outcomes~\cite{Chen:PRA:2013}, or both~\cite{Meng:PRA:2018}. As with~\cite{Meng:PRA:2018}, in this work, we propose two generalizations of Hardy's arguments applicable to an arbitrary (symmetric) bipartite Bell scenario, which recovers, respectively, that of~\cite{Meng:PRA:2018} and that discussed by~\cite{Liang:2005vl,KCA+06,Cereceda16} as a special case. We provide evidence showing that the one that can be interpreted as a demonstration of FTI leads to a degree of success higher than all those offered by other existing proposals, even in the presence of  imperfections.

\section{Hardy's paradox and its generalization in the simplest Bell scenario }\label{Sec:Hardy:CHSH}

\subsection{Hardy's original formulation}

Consider the simplest Clauser-Horne-Shimony-Holt (CHSH) Bell scenario, i.e., one in which two observers each perform two binary-outcome measurements. 
Let $x$ and $y$ ($a$ and $b$) represent, respectively, the setting/ input (outcome/ output) of Alice and Bob side, and $A_x$ ($B_y$) $=0,1$ denotes the outcome of Alice (Bob) when given input $x$ ($y$) $=0,1$. The probability distribution $\{P(a,b|x,y)=P(A_x,B_y)\}$ admissible in LHV models can be described by convex mixtures of local deterministic strategies $\{A_x=f_A(x,\lambda),B_y=f_B(y,\lambda)\}$, where $f_A$ ($f_B$) is a deterministic function 
of the input $x$ ($y$) and LHV $\lambda$. The Hardy paradox of~\cite{Hardy93} is encapsulated by:
\begin{equation}\label{HardyParadox}
\begin{split}
     &P(0,0|0,0) = 0, \quad
    P(1,1|0,1) = 0,\\
    &P(1,1|1,0) = 0,\quad
    P(1,1|1,1) = q> 0. 
\end{split}
\end{equation}

For DLHVMs, the equality constraints of \cref{HardyParadox} imply $P(1,1|1,1)=0$, which contradicts the inequality constraint of \cref{HardyParadox}. In other words, together with the equality constraints, the occurrence of the event $x=y=a=b=1$ contradicts the prediction of {\em any} DLHVM. Consequently, the quantity $q\equiv P(1,1|1,1)$ is also known as the {\em success probability}.  In quantum theory, it is known~\cite{Rabelo12} that the maximal attainable success probability is $\frac{5\sqrt{5}-11}{2}\approx 9.02\%$. Finally, note that a general LHV model can always be seen mathematically as a convex mixture of DLHVM. Thus, the observation of \cref{HardyParadox} also rules out a general LHV model.

\subsection{Generalization due to Cabello-Liang-Li}

Cabello's~\cite{Cabello02} relaxation of Hardy's argument, originally proposed for a tripartite scenario and subsequently applied in the bipartite scenario by Liang and Li~\cite{Liang:2005vl}, takes the form:
 \begin{equation}\label{CabelloArgument}
 \begin{split}
     P(0,0|0,0) = p,\quad
    P(1,1|0,1) = 0,\\
    P(1,1|1,0) = 0,\quad
    P(1,1|1,1) = q.
    \end{split}
\end{equation}
Hereafter, we refer to this as the Cabello-Liang-Li (CLL) argument. Compared with~\cref{HardyParadox}, we see that in this argument, $P(0,0|0,0)$ is allowed to take nonzero value. From~\cref{Fig:Hardy}, we see that for any DLHVM, if the event $x=y=a=b=1$ occurs, so must the event $x=y=a=b=0$. However, there exist other local deterministic strategies (e.g., one where $a=b=0$ regardless of $x$ and $y$) where the latter event occurs while the former does not. Thus, for the prediction of a general LHV model, we must have $p\equiv P(0,0|0,0) \ge q \equiv P(1,1|1,1)$. In other words, one may take the positive value of the quantity 
\begin{equation}\label{Eq:Success:Cabello}
	q-p = P(1,1|1,1)-P(0,0|0,0)
\end{equation}
as a witness for successfully demonstrating a logical contradiction based on such an argument. In~\cite{KCA+06}, the authors refer to $q-p$ as the success probability of such an argument. However, since $q-p$ is the difference between two conditional probabilities, we shall follow~\cite{Rai:PRA:2021} and refer to $q-p$ instead as the {\em degree of success} (DS). In~\cite{KCA+06}, the authors showed that this  DS can reach $\approx10.79\%$.

\subsection{Our generalization based on FTI}

 Now, let us revisit \cref{HardyParadox} and see how a nonzero value of $q$ can be understood as a failure of the transitivity of implications (FTI), thanks to Stapp's formulation of Hardy's paradox in~\cite{Stapp93}. To this end, note from the two zero-constraints on the left of \cref{HardyParadox} that they entail not only the inferences of~\cref{Fig:Hardy}, but also those of~\cref{Fig:Stapp}, i.e., 
\begin{equation}
	A_1=1\implies B_0=0\implies A_0=1
\end{equation}
Thus, if $q>0$, meaning that the event $A_1=1$ and $B_1=1$ can simultaneously occur with some nonzero probability, and {\em if} implications are {\em transitive} (as in classical logic), it must be the case that, at least {\em sometimes} $B_1=1\implies  A_0=1$. However, this contradicts the remaining zero constraint in~\cref{HardyParadox}, thus manifesting an FTI.

Using this reformulation, we now provide a different relaxation of Hardy's paradox via:
 \begin{equation}\label{FTIArgument}
 \begin{split}
     P(0,0|0,0) = 0,\quad
    P(1,1|0,1) = r,\\
    P(1,1|1,0) = 0,\quad
    P(1,1|1,1) = q,
    \end{split}
\end{equation}
From \cref{Fig:Stapp}, we see that for any DLHVMs, if the event $x=y=a=b=1$ occurs, so must the event for $x=0$ and $y=a=b=1$. However, there are local deterministic strategies where the converse does not hold. Thus, for a general LHV model (obtained by averaging the deterministic ones), we must have $r\equiv P(1,1|0,1) \ge q \equiv P(1,1|1,1)$. Hence, in analogy with the CLL argument, we refer to
\begin{equation}\label{Eq:Success:FTI}
	q-r = P(1,1|1,1)-P(1,1|0,1)
\end{equation}
as the  DS of such an argument. In~\cite{Chen23}, it was shown that in quantum theory, the largest value of $q-r$ is $\frac{1}{8}=12.5\%$, attainable by performing projective measurements on a two-qubit pure state and higher than that achievable with \cref{CabelloArgument}.

\subsubsection{Maximal degree of success for $2$-qubit pure states}

In fact, the general two-qubit {\em pure} state and observables satisfying the zero-probability constraints of \cref{FTIArgument} are~\cite{Chen23}:
\begin{subequations}\label{Eq:Class2b}
\begin{gather}\label{Eq:HardyState}
    \ket{\psi} = \sin\theta\left(\cos{\alpha}\ket{0}-\sin\alpha\ket{1}\right)\!\ket{1}+\cos\theta\ket{1}\!\ket{0},\\
\begin{split}\label{Eq:2b:Observables}	
    A_0=\sigma_z,  \quad
    A_1=\cos{2\alpha}\,\sigma_z-\sin{2\alpha}\,\sigma_x, \\
    B_0=\sigma_z, \quad
    B_1=\cos{2\beta}\,\sigma_z-\sin{2\beta}\,\sigma_x,
\end{split}
\end{gather}
\end{subequations}
where $\theta,\alpha,\beta\in[0,\pi]$. From here, the corresponding  DS of \cref{Eq:Success:FTI,Eq:Class2b}, as a function of the parameters $\theta$, $\alpha$, and $\beta$ can be shown to be:
\begin{equation}\label{Success_Probability}
\begin{split}
    P_{\text{succ}}(\theta,\alpha,\beta) = &\frac{1}{2}\sin{\alpha}\left[\left(\cos{2\beta}\cos{2\theta}-1\right)\sin{\alpha}\right.\\
    &\left.+\sin{2\beta}\sin{2\theta}\right].
\end{split}
\end{equation}
Naturally, one may wonder which entangled two-qubit pure state gives the largest value of $P_{\text{succ}}$. To this end, note that the entanglement of the two-qubit state of \cref{Eq:HardyState}, as measured according to the concurrence \cite{Wootters:PRL:1998}, is
\begin{equation}\label{Eq:Concurrence_Constraint}
    C(\ket{\psi}) = |\sin{2\theta}\cos{\alpha}|.
\end{equation}
Using variational techniques, the largest DS that we have found for given concurrence $C$ is
\begin{equation}\label{Eq:Psucc_FixedTheta}
    P^{*}_{\text{succ}}(C)=\frac{\sqrt{1-C^2}}{2}\left(1-\sqrt{1-C^2}\right),
\end{equation}
for which $\theta=\beta\in\{\frac{\pi}{4},\frac{3\pi}{4}\}$.
From~\cref{fig:Compare_state}, it is clear that for any given concurrence, this DS is always larger that from the CLL argument, which, in turn, is larger than that from Hardy's original formulation. For completeness, we include in \cref{App:DSvsC} the parametric form of the maximal DS as a function of the concurrence for the CLL argument and Hardy's original formulation.

\begin{figure}[h!tbp]
\captionsetup{justification=RaggedRight,singlelinecheck=off}
\centering
  \includegraphics[width=\linewidth]{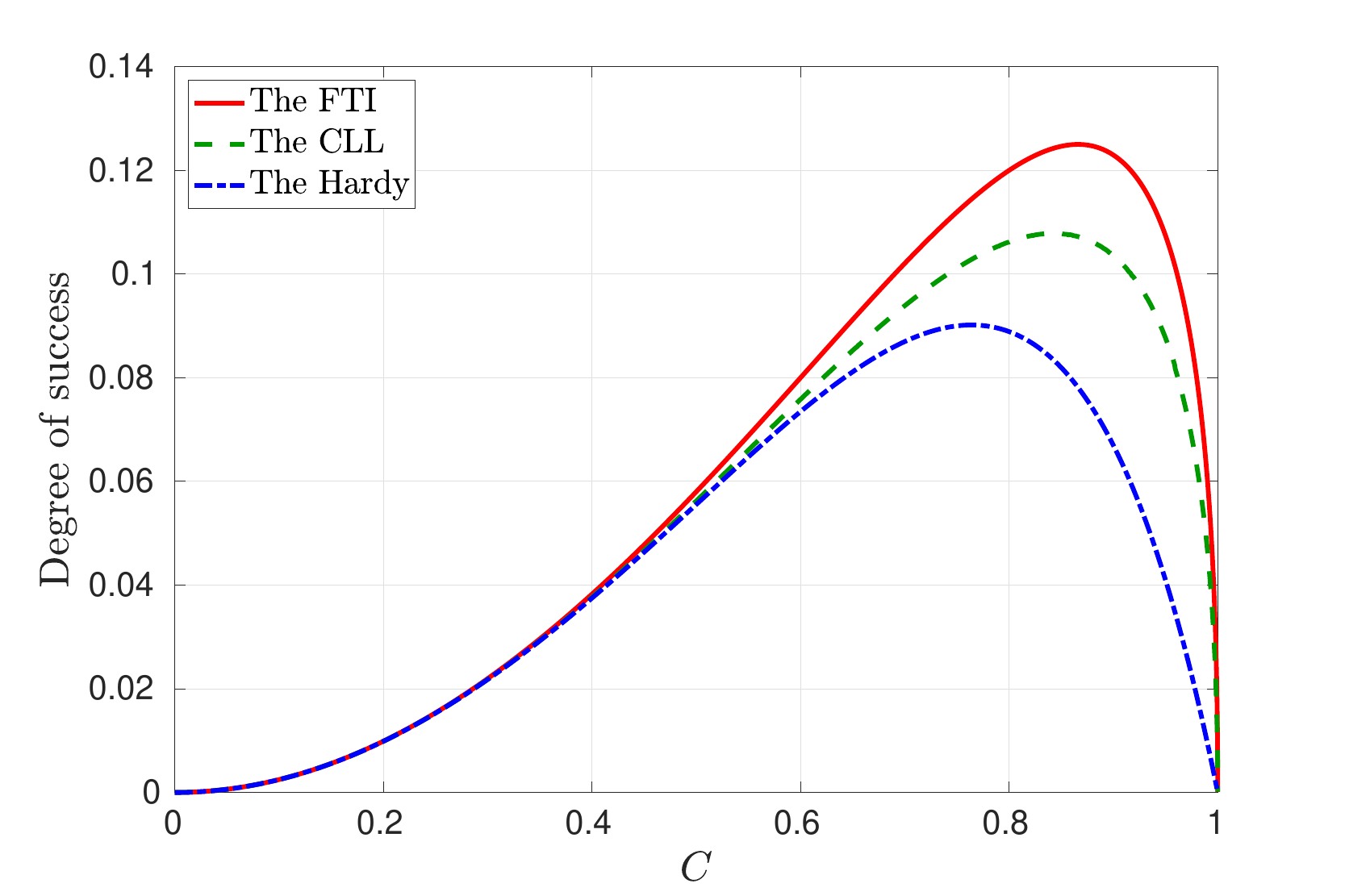}
\caption{The maximal DS (which reduces to the success probability in the Hardy argument) in demonstrating a proof of nonlocality without inequality for given concurrence~\cite{Wootters:PRL:1998}.
From top to bottom, we plot $q-r$ of~\cref{FTIArgument} for our FTI argument (red, solid), $q-p$ of~\cref{CabelloArgument} for CLL's argument (green, dashed, see~\cite{KCA+06}), and $q$ of Eq.~\eqref{HardyParadox} for Hardy's argument (blue, dashed-dotted, see~\cite{Hardy93}). Note that the entangled state that gives the largest DS differs from one formulation to the other.}
\label{fig:Compare_state}
\end{figure}

\subsubsection{FTI argument with  imperfections}

Evidently, imperfections in any realistic experimental scenario make it essentially impossible to realize the zero-probability equality constraints in all these different formulations. To understand the impact of these imperfections, we now {\em relax}~\cref{FTIArgument} and consider
\begin{equation}\label{Eq:noisyFTI}
\begin{split}
        P(0,0|0,0) \le \epsilon,\quad
    P(1,1|0,1) = r,\\
    P(1,1|1,0) \le \epsilon,\quad
    P(1,1|1,1) = q,   
\end{split}
\end{equation}
where $\epsilon$ is the tolerance from a deviation of the zero-probability equality constraints.\footnote{Note that $\epsilon$ should {\em not} be understood as an uncertainty in the estimation of the conditional probability; otherwise, similar tolerance should also be included in the expression for $P(1,1|0,1)$ and $P(1,1|1,1)$.}
Then, the maximal DS allowed in an LHV theory satisfies:
\begin{equation} \label{Eq:noisyFTI_LocalBound}
    q-r -2\epsilon= P(1,1|1,1) - P(1,1|0,1)-2\epsilon\overset{\L}{\le} 0,
\end{equation}
which follows from the following rewriting~\cite{Rai:PRA:2021,Rabelo12} of the Clauser-Horne Bell inequality~\cite{Clauser74}:
\begin{equation}\label{Ineq:CH}
    P(1,1|1,1)-P(1,1|1,0)-P(1,1|0,1)-P(0,0|0,0)\overset{\L}{\le} 0.
\end{equation}
From \cref{Eq:noisyFTI_LocalBound}, we see that when an $\epsilon$-deviation from the zero-probability constraints is allowed, it is expedient to consider, instead, $q-r-2\epsilon$ as the {\em generalized} DS. That is, whenever this quantity is nonzero, we again find a contradiction with the prediction given by all LHV theories.

For the CLL argument, a similar discussion has been made in \cite{Rai:PRA:2021}, giving rise to a generalized DS of $q-p-2\epsilon$, see~\cref{CabelloArgument}. In \cref{fig:Compare_error}, we show the corresponding maximal generalized DS (MGDS), i.e., the largest value of $q-r-2\epsilon$ from~\cref{Eq:noisyFTI} and $q-p-2\epsilon$ from~\cref{CabelloArgument} when the tolerance $\epsilon\in (0,0.5)$. Our results clearly show that for any amount of tolerance in this range, the FTI argument generally gives a somewhat higher MGDS  compared to the CLL one. Moreover, as we see from~\cref{fig:Compare_error}, the MGDS for both arguments increases when $\epsilon$ increases from zero up until some critical value. Qualitatively, we can appreciate this by noting that when the zero-probability constraints are slightly relaxed, the range of available nonlocal quantum strategies also increases. Nonetheless, when $\epsilon$ is sufficiently large, we see  from \cref{Eq:noisyFTI_LocalBound} that a nonzero generalized DS must involve a $P(1,1|1,1)$ close to unity and a $P(1,1|0,1)$ close to 0, i.e., the corresponding $\vecP$ must become close to that producible by a DLHVM, thereby resulting in a decrease in MGDS.

Note that for our relaxed FTI argument of \cref{Eq:noisyFTI}, upper bounds on the MGDS (from level-3 of the semidefinite programming (SDP) hierarchy introduced in \cite{Moroder13}) coincide with the lower bounds (based on a two-qubit pure state with rank-1 projective measurements) to within a numerical precision of $10^{-7}$. However, for the CLL argument, if we consider only qubit strategies, then as already noted in \cite[Fig.~1]{Rai:PRA:2021},  there appears to be a gap between the MGDS achievable with such strategies and the SDP upper bound (for $\epsilon_1\lesssim \epsilon < \epsilon_2$ where $\epsilon_1=0.158, \epsilon_2=0.5$). Upon closer inspection, we find that for every $\epsilon$ in this interval, the SDP upper bound on MGDS is attainable to within the same precision by considering an appropriate convex mixture of the qubit strategy for  $\epsilon = \epsilon_1$ and $\epsilon = \epsilon_2$, or equivalently a ququart strategy obtained from their direct sum. 
That is, we can saturate the SDP upper bound for $\epsilon=p\epsilon_1+(1-p)\epsilon_2$ by mixing the quantum strategies for $\epsilon = \epsilon_1$ and $\epsilon = \epsilon_2$, respectively, with weight $p$ and $1-p$ while fulfilling all other constraints, cf.~\cref{CabelloArgument} with tolerance $\epsilon$.

\begin{figure}[h!tbp]
\captionsetup{justification=RaggedRight,singlelinecheck=off}
\centering
  \includegraphics[width=\linewidth]{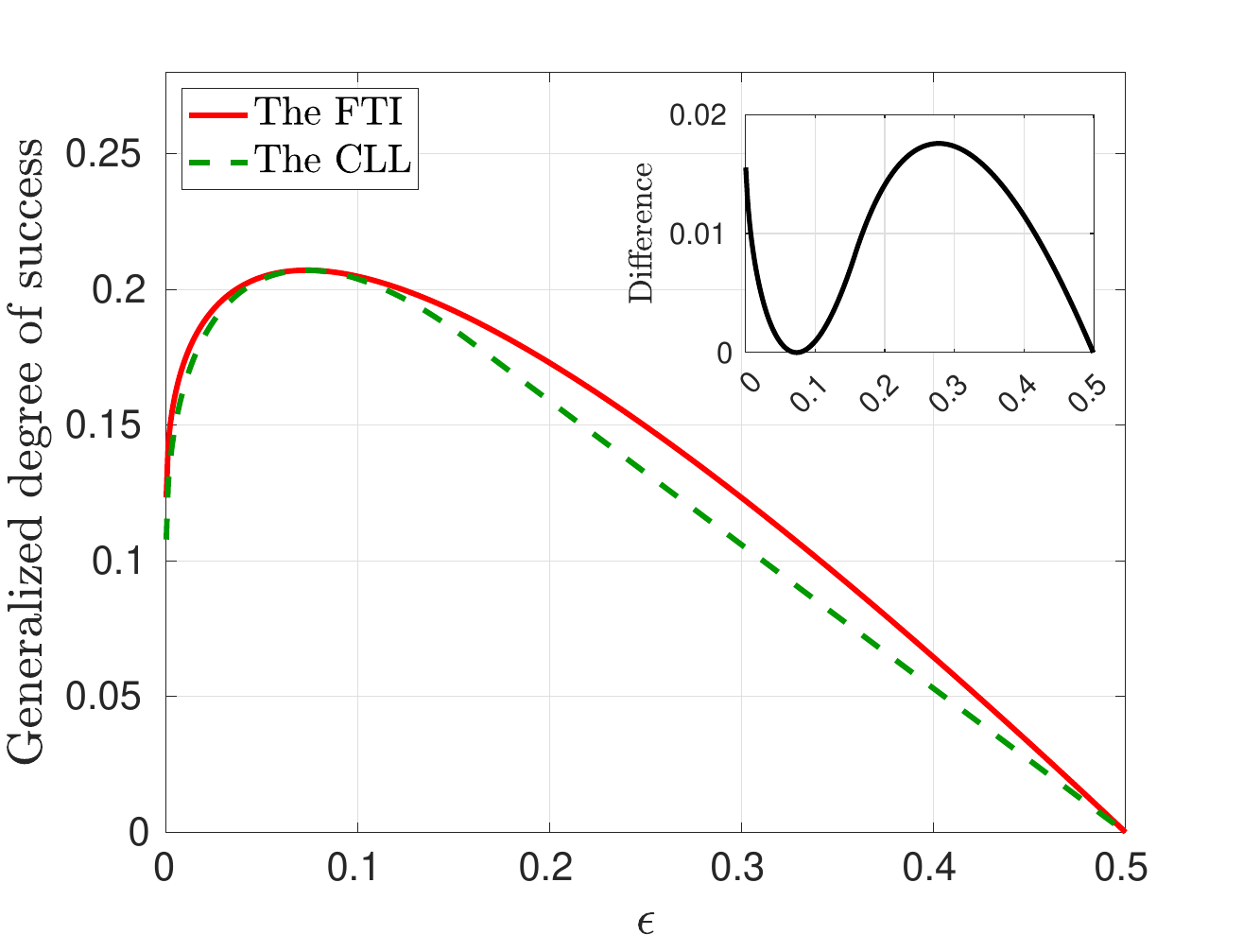}
\caption{Maximal generalized degree of success (MGDS) as a function of the deviation $\epsilon$ from the zero-probability constraints. The top curve (red, solid) shows a lower bound on the MGDS for our FTI argument (obtained by optimizing over two-qubit states and projective measurements); the middle curve (green, dashed) shows an upper bound on the MGDS for CLL's argument (see also \cite{Rai:PRA:2021}) obtained using level-3 of the semidefinite programming (SDP) hierarchy described in Ref.~\cite{Moroder13}. The inset highlights the difference between the two MGDS (provided by FTI's MGDS less the CLL's), which is always positive within a numerical precision of $10^{-7}$. Note that two probabilities coincide when $\epsilon \approx 0.0730$.
}
\label{fig:Compare_error}
\end{figure}

\section{Generalization of Hardy's proof beyond the simplest Bell scenario }\label{Sec:Hardy}

\begin{figure*}[h!tbp]
\captionsetup{justification=RaggedRight,singlelinecheck=off}
\centering
\begin{tikzpicture}
\node[rectangle,draw,minimum width = 1.5cm, 
    minimum height = 1.5cm] at (-1.9,5) {The ladder proof \cite{Boschi97}};
    \node[rectangle,draw,minimum width = 1.5cm, 
    minimum height = 1.5cm] at (-6,5) {Cereceda \cite{Cereceda16}};
    \node[rectangle,draw,minimum width = 1.5cm, 
    minimum height = 1.5cm] at (0,8) {Meng {\em et al.} \cite{Meng:PRA:2018}};
    \node[rectangle,draw,minimum width = 1.5cm, 
    minimum height = 1.5cm] at (1.4,5) {Chen {\em et al.} \cite{Chen:PRA:2013}};
    \node[rectangle,draw,minimum width = 1.5cm, 
    minimum height = 1.5cm] at (-6,8) {\bf Generalized CLL [\cref{Extend_Hardy}]};
    \node[rectangle,draw,minimum width = 1.5cm, 
    minimum height = 1.5cm] at (6,8) {\bf Generalized FTI [\cref{Extend_Stapp}]};
    \node[rectangle,draw,minimum width = 1.5cm, 
    minimum height = 1.5cm] at (-6,2) {Cabello-Liang-Li (CLL)~\cite{Liang03,Liang:2005vl,Cabello02}};
    \node[rectangle,draw,minimum width = 1.5cm, 
    minimum height = 1.5cm] at (0,2) {Hardy's paradox \cite{Hardy93}};
    \node[rectangle,draw,minimum width = 1.5cm, 
    minimum height = 1.5cm] at (6,2) {\bf The FTI [\cref{FTIArgument}]};
    \node[] at (-3.1,8.3) { $p=0$};
    \node[] at (-4.3,5.3) { $p=0$};
    \node[] at (3.2,8.3) { $r=0$};
    \node[] at (-5.5,6.5) { $d=2$};
    \node[] at (-5.5,3.5) { $k=2$};
    \node[] at (-1.6,6.5) { $d=2$};
    \node[] at (1.6,6.5) { $k=2$};
    \node[] at (7,5) { $d=k=2$};
    \node[] at (-1.6,3.5) { $k=2$};
    \node[] at (1.6,3.5) { $d=2$};
    \node[] at (-2.8,2.3) { $p=0$};
    \node[] at (3.2,2.3) { $r=0$};
    \draw[-Triangle,line width=1pt] (-3.8,8)--(-1.3,8);
    \draw[-Triangle,line width=1pt] (4,8)--(1.3,8);
    \draw[-Triangle,line width=1pt] (-6,7)--(-6,6);
    \draw[-Triangle,line width=1pt] (-1,7)--(-1,6);
    \draw[-Triangle,line width=1pt] (1,7)--(1,6);
    \draw[-Triangle,line width=1pt] (-4.8,5)--(-3.6,5);
    \draw[-Triangle,line width=1pt] (-6,4)--(-6,3);
    \draw[-Triangle,line width=1pt] (-1,4)--(-1,3);
    \draw[-Triangle,line width=1pt] (1,4)--(1,3);
    \draw[-Triangle,line width=1pt] (6,7)--(6,3);
    \draw[-Triangle,line width=1pt] (4.6,2)--(1.6,2);
    \draw[-Triangle,line width=1pt] (-3.7,2)--(-1.6,2);
\end{tikzpicture}
\caption{\label{Fig:Summary}Summary of the relationships between the various Hardy-type paradoxes discussed in this work. Our new constructions are printed in boldface.}
\end{figure*}
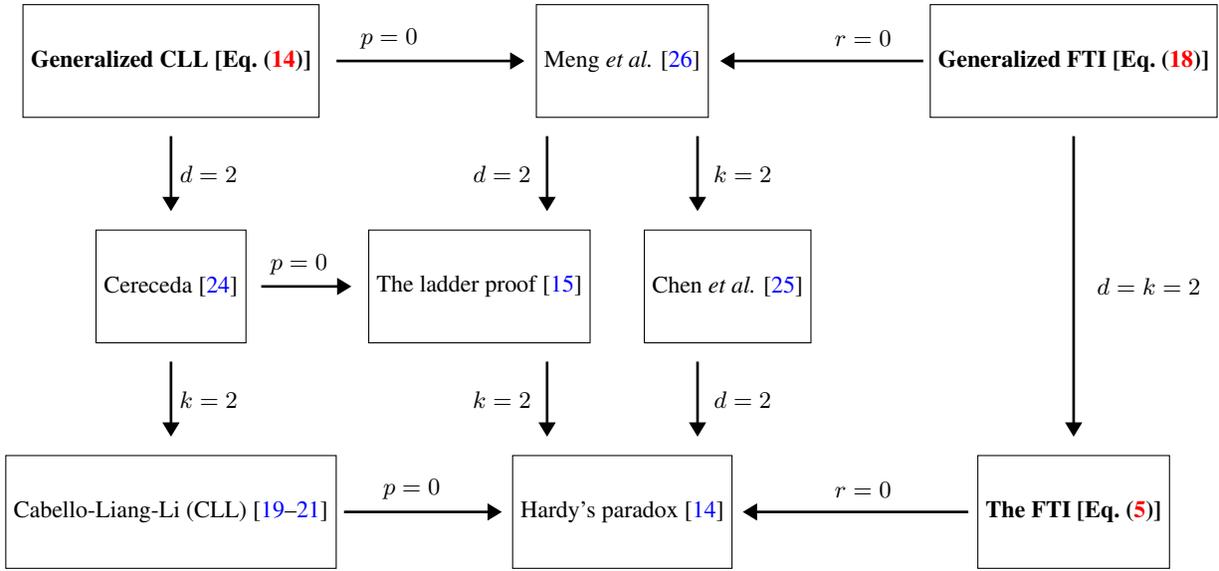

Having understood how Hardy and Hardy-type paradoxes work in the simplest Bell scenario, the time is now ripe to discuss their generalization to more complex Bell scenarios. In this section, we propose, respectively, a generalization of both the Hardy-type paradox of CLL, \cref{CabelloArgument}, and that based on FTI, \cref{FTIArgument}, to an {\em arbitrary} bipartite $k$-input $d$-output Bell scenario, i.e., one in which both party has a choice over $k$ alternative $d$-outcome measurements. 

\subsection{Generalization of CLL Hardy-type paradox}

Specifically, for the CLL Hardy-type paradox, the following conditions on the joint conditional probabilities:
\begin{subequations}\label{Extend_Hardy}
    \begin{align}
        &\left\{\begin{aligned}\label{Extend_Hardy_a}
        &P(A_{k-1}<B_{k-1})=q,~\text{if}~k\in \text{odd},\\
        &P(A_{k-1}>B_{k-1})=q,~\text{if}~k\in \text{even},
        \end{aligned}
        \right .\\
        &P(A_i>B_{i-1})=0,~\forall~ i\in \text{odd}\cap \{1,\cdots,k-1\},\label{Extend_Hardy_b}\\
        &P(A_i<B_{i-1})=0,~\forall~ i\in \text{even}\cap \{1,\cdots,k-1\},\label{Extend_Hardy_c}\\
        &P(A_{i-1}>B_i)=0,~\forall~ i\in \text{odd}\cap \{1,\cdots,k-1\},\label{Extend_Hardy_d}\\
        &P(A_{i-1}<B_i)=0,~\forall~ i\in \text{even}\cap \{1,\cdots,k-1\},\label{Extend_Hardy_e}\\
        &P(A_0<B_0)=p,\label{Extend_Hardy_f}
    \end{align}
\end{subequations}
together with $q>p$ define our generalization of this paradox, where the outcomes $A_x, B_y$ may take $d$ possible values, say, from $\{0,1,\cdots, d-1\}$. For the special case of $p=0$, one obtains a generalization of the original Hardy paradox to an {\em arbitrary} bipartite $k$-input $d$-output Bell scenario. If we further set $d=2$, then the construction reduces to one equivalent (under relabeling of inputs and outputs) to the ladder proof of nonlocality~\cite{Boschi97}. If, instead, we take $d=2$ in \cref{Extend_Hardy} without setting $p=0$, one obtains the argument briefly discussed in~\cite{Cereceda16}. All these relations are summarized in~\cref{Fig:Summary}.

\begin{figure}[h!tbp]
\captionsetup{justification=RaggedRight,singlelinecheck=off}
\centering
\begin{tikzpicture}
    \node[] at (0,5.5) {$A_{k-1}=s^A_{k-1}$};
    \node[] at (4,5.5) {$B_{k-1}=s^B_{k-1}$};
    \node[] at (0,4.25) {$A_{k-2}=s^A_{k-2}$};
    \node[] at (4,4.25) {$B_{k-2}=s^B_{k-2}$};
    \node[] at (0.2,1.25) {$A_1=s^A_{1}$};
    \node[] at (3.8,1.25) {$B_1=s^B_{1}$};
    \node[] at (0.2,0) {$A_0=s^A_{0}$};
    \node[] at (3.8,0) {$B_0=s^B_{0}$};
    \node[blue] at (2,5.75) {$q$};
    \node[red] at (2,-0.25) {$p$};
    \draw[dotted,line width=1pt] (0.5,2.25)--(0.5,3.25);
    \draw[dotted,line width=1pt] (3.5,2.25)--(3.5,3.25);
    \draw[dotted,line width=1pt] (2,2.5)--(2,3);
    \draw[Triangle-Triangle,blue,line width=1pt] (1,5.5)--(3,5.5);
    \draw[-Triangle,line width=1pt] (1,5.4)--(3,4.3);
    \draw[-Triangle,line width=1pt] (3,5.4)--(1,4.3);
    \draw[-Triangle,line width=1pt] (1,4.15)--(3,3.05);
    \draw[-Triangle,line width=1pt] (3,4.15)--(1,3.05);
    \draw[-Triangle,line width=1pt] (1,2.4)--(3,1.3);
    \draw[-Triangle,line width=1pt] (3,2.4)--(1,1.3);
    \draw[-Triangle,line width=1pt] (1,1.15)--(3,0.1);
    \draw[-Triangle,line width=1pt] (3,1.15)--(1,0.1);
    \draw[Triangle-Triangle,red,line width=1pt] (1,-0.075)--(3,-0.075);
\end{tikzpicture}
\caption{Logical structure of the generalized CLL arguments in the $k$-input $d$-output Bell scenario. The generalization of Hardy's original argument to these cases is recovered by setting $p=0$.}
\label{Figure_Extend_Hardy}
\end{figure}
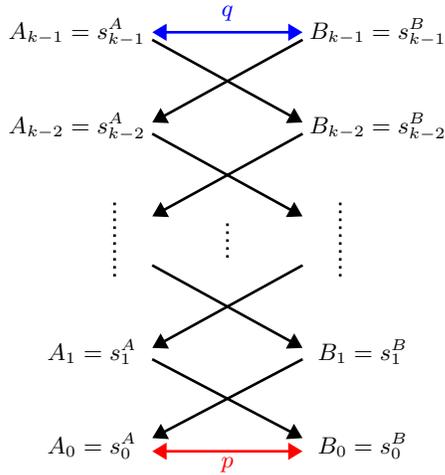

To see that the constraints of \cref{Extend_Hardy} with $q-p>0$ indeed constitute a proof of nonlocality without inequality, we begin by 
restricting our attention to a DLHVM where the measurement outcomes take definite values, denoted by $\{A_{i}=s^A_{i}\}$ and $\{B_{i}=s^B_{i}\}$, where $i\in\{0,1,\dots,k-1\}$. We depict the logical structure behind this argument schematically in \cref{Figure_Extend_Hardy}. 
Let us now consider the case of even and odd $k$ separately, starting with odd $k$. Then, in order for a DLHVM to reproduce~\cref{Extend_Hardy_a}, i.e., $P(A_{k-1}<B_{k-1})=q> 0$, the model must produce events $\{A_{i}=s^A_{i}\}$ and $\{B_{i}=s^B_{i}\}$ such that $s^A_{k-1}<s^B_{k-1}$. Similarly, the other constraints of \cref{Extend_Hardy} imply constraints on the relationship between $\{s^A_{i}\}$ and $\{s^B_{i}\}$, where $i\in\{0,1,\dots,k-1\}$. For example, together with the  conditions of \cref{Extend_Hardy_c} and \cref{Extend_Hardy_e} for $i=k-1$, i.e., $P(A_{k-1}<B_{k-2})=0$ and $P(A_{k-2}<B_{k-1})=0$, we get 
\begin{equation}
	s^B_{k-2}\le s^A_{k-1}<s^B_{k-1} \le s^A_{k-2}. 
\end{equation}
By considering the other zero-probability constraints one at a time for the remaining $i$, we arrive at
\begin{equation}\label{Chain_Extend_Hardy_odd}
    s^A_{0}\le \dots \le s^B_{k-2}\le s^A_{k-1}<s^B_{k-1} \le s^A_{k-2}\le \dots\le  s^B_{0}.
\end{equation}
This means that for any DLHVMs that give $s^A_{k-1}<s^B_{k-1}$, the constraints of \cref{Chain_Extend_Hardy_odd} imply that they must also give $s^A_0<s^B_0$. However, there can be other DLHVMs where $s^A_0<s^B_0$ holds even though $s^A_{k-1}\not<s^B_{k-1}$. Thus, for a general LHV model, the conditions of \cref{Chain_Extend_Hardy_odd} imply that $p\ge q$. In other words, a nonzero value of the DS $q-p$ witness Bell nonlocality without resorting to a Bell inequality. 

Similarly, for even $k$, starting from $P(A_{k-1}>B_{k-1})=q>0$ and by considering the other zero-probability constraints lead to, for any DLHVMs,
\begin{equation}\label{Chain_Extend_Hardy_even}
    s^A_{0}\le \dots \le s^A_{k-2}\le s^B_{k-1}<s^A_{k-1} \le s^B_{k-2}\le \dots\le  s^B_{0}.
\end{equation}
Again, this observation implies that $p\ge q\Leftrightarrow p-q\ge 0$ for any LHV model for all $k\ge 2$ and $d \ge 2$. 

\subsection{Generalization of the FTI-based Hardy-type paradox}

Next, let us describe our generalization of the FTI-based Hardy-type paradox from \cref{FTIArgument}, which consists of the following conditions:
\begin{subequations}\label{Extend_Stapp}
    \begin{align}
        &P(A_{k-1}<B_{k-1})=q,\label{Extend_Stapp_a}\\
        &P(A_i<B_{i-1})=0,~\forall~i\in \{1,\cdots,k-1\},\label{Extend_Stapp_b}\\
        &P(B_{i-1}<A_{i-1})=0,~\forall~ i\in \{1,\cdots,k-1\},\label{Extend_Stapp_c}\\
        &P(A_0<B_{k-1})=r\label{Extend_Stapp_d},
    \end{align}
\end{subequations}
and the requirement of $q>r$. The special case of $r=0$, which can be seen as a generalization of Stapp's argument~\cite{Stapp93}, has been proposed and discussed in \cite{Meng:PRA:2018}. To recover~\cref{FTIArgument} from \cref{Extend_Stapp}, one sets $k=d=2$ and apply the relabeling $A_i=0\leftrightarrow A_i=1$ for all $i\in\{0,1\}$. In \cref{Figure_Extend_Stapp}, we depict schematically the logical structure of this paradox.

As with our  explanation to \cref{Extend_Hardy}, for any  DLHVM satisfying $P(A_{k-1}<B_{k-1})=q>0$, the model must produce events $\{A_{i}=s^A_{i}\}$ and $\{B_{i}=s^B_{i}\}$ such that $s^A_{k-1}<s^B_{k-1}$. At the same time, the other inequality constraints from \cref{Extend_Stapp} imply $s^B_{k-2}\le s^A_{k-1}$, $s^A_{k-2}\le s^B_{k-2}$, etc., leading to
\begin{equation}\label{Chain_Extend_FTI}
    s^A_{0}\le s^B_{0}\le \dots \le s^A_{k-2}\le s^B_{k-2}\le s^A_{k-1} < s^B_{k-1},
\end{equation}
which implies $s^A_{0} < s^B_{k-1}$.
This means that with the zero-probability constraints, a DLHVM equipped with a strategy giving $s^A_{k-1}<s^B_{k-1}$ must also give $s^A_{0} < s^B_{k-1}$. Again, other DLHVM may give $s^A_{0} < s^B_{k-1}$ even though $s^A_{k-1}\not<s^B_{k-1}$. Thus, from \cref{Extend_Stapp_d}, we conclude that for a general LHV model (obtained by averaging over local deterministic strategies), we must have $r \ge q \Leftrightarrow r-q \ge 0$
for all $k, d\ge 2$.

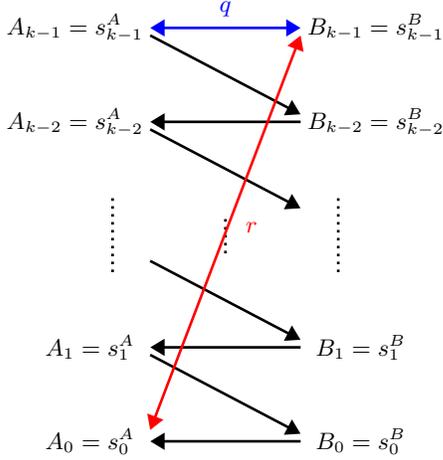
\begin{figure}[h!tbp]
\captionsetup{justification=RaggedRight,singlelinecheck=off}
\centering
\begin{tikzpicture}
    \node[] at (0,5.5) {$A_{k-1}=s^A_{k-1}$};
    \node[] at (4,5.5) {$B_{k-1}=s^B_{k-1}$};
    \node[] at (0,4.25) {$A_{k-2}=s^A_{k-2}$};
    \node[] at (4,4.25) {$B_{k-2}=s^B_{k-2}$};
    \node[] at (0.2,1.25) {$A_1=s^A_{1}$};
    \node[] at (3.8,1.25) {$B_1=s^B_{1}$};
    \node[] at (0.2,0) {$A_0=s^A_{0}$};
    \node[] at (3.8,0) {$B_0=s^B_{0}$};
    \node[blue] at (2,5.75) {$q$};
    \node[red] at (2.35,2.85) {$r$};
    \draw[dotted,line width=1pt] (0.5,2.25)--(0.5,3.25);
    \draw[dotted,line width=1pt] (3.5,2.25)--(3.5,3.25);
    \draw[dotted,line width=1pt] (2,2.5)--(2,3);
    \draw[Triangle-Triangle,blue,line width=1pt] (1,5.5)--(3,5.5);
    \draw[-Triangle,line width=1pt] (1,5.4)--(3,4.35);
    \draw[-Triangle,line width=1pt] (3,4.25)--(1,4.25);
    \draw[-Triangle,line width=1pt] (1,4.15)--(3,3.1);
    \draw[-Triangle,line width=1pt] (1,2.4)--(3,1.35);
    \draw[-Triangle,line width=1pt] (3,1.25)--(1,1.25);
    \draw[-Triangle,line width=1pt] (1,1.15)--(3,0.1);
    \draw[-Triangle,line width=1pt] (3,0)--(1,0);
    \draw[Triangle-Triangle,red,line width=1pt] (1,0.15)--(3,5.4);
\end{tikzpicture}
\caption{Logical structure of our FTI arguments in the $k$-input $d$-output Bell scenario. Generalized Stapp's arguments as introduced in~\cite{Meng:PRA:2018} are recovered by setting $r=0$.}
\label{Figure_Extend_Stapp}
\end{figure}

\subsection{Proof of equivalence of generalized Stapp's proof and generalized ladder proof of nonlocality}

Interestingly, the authors of \cite{Meng:PRA:2018} proved that in the $k$-input $2$-output scenario, the generalized Stapp's argument, \cref{Extend_Stapp} with $r=0$, and the ladder proof of nonlocality, cf. \cref{Extend_Hardy} with $p=0$, are equivalent. In other words, these two sets of conditions can be obtained from each other via an appropriate relabeling of inputs and outputs. In what follows, we show that this equivalence also holds for arbitrary $k,d\ge 2$.

\begin{theorem}\label{Thm:Equivalence}
For any symmetric bipartite Bell scenario, the set of conditions given in~\cref{Extend_Hardy} with $p=0$ is equivalent to the set of conditions of~\cref{Extend_Stapp} with $r = 0$.
\end{theorem}

\begin{proof}
Let us rewrite Alice's and Bob's measurement outcomes in~\cref{Extend_Hardy}, respectively, as $A'_{x'}$ and $B'_{y'}$, and let 
\begin{gather}\label{Eq:IntevalSets}
	\mathscr{L}_{1k-}:=\{1,2,\ldots, \lfloor\tfrac{k}{2}\rfloor-1\},\quad \mathscr{L}_{1k}:=\mathscr{L}_{1k-}\cup\{\lfloor\tfrac{k}{2}\rfloor\},\nonumber\\
	\mathscr{L}_{0k-}:=\{0\}\cup\mathscr{L}_{1k-},\quad \mathscr{L}_{0k}:=\mathscr{L}_{0k-}\cup\{\lfloor\tfrac{k}{2}\rfloor\}.
\end{gather}	

For {\em odd} $k$, one may verify that the following relabeling 
\begin{equation}
\begin{split}
	A'_{x'}=\left\{
	\begin{array}{ll}
	A_{\ell-1}, & x' = k-2\ell~\land~ \ell\in\mathscr{L}_{1k}\\
	A_{k-1-\ell}, & x' = k-1-2\ell~\land~\ell\in\mathscr{L}_{0k}
	\end{array}\right.,\\
	B'_{y'}=\left\{
	\begin{array}{ll}
	B_{k-1}, & y'=k-1\\
	B_{k-1-\ell}, & y' = k-2\ell~\land~\ell\in\mathscr{L}_{1k}\\
	B_{\ell-1}, & y' = k-1-2\ell~\land~ \ell\in\mathscr{L}_{1k}
	\end{array}\right.,\\
\end{split}
\end{equation}
transforms the $2k$ conditions of~\cref{Extend_Hardy} to those of~\cref{Extend_Stapp}. To see this, note that under this transformation, the condition of~\cref{Extend_Hardy_a} stays as
\begin{equation}
	P(A'_{k-1}<B'_{k-1}) = P(A_{k-1}<B_{k-1})=q, 
\end{equation}
which is~\cref{Extend_Stapp_a}. With the transformation, the conditions of~\cref{Extend_Hardy_b},~\cref{Extend_Hardy_d},~\cref{Extend_Hardy_c}, and~\cref{Extend_Hardy_f} with $p=0$, respectively, become the requirements that each of the following probabilities vanish:
\begin{align}\label{Eq:Transformed}
	&P(A'_{k-2\ell}>B'_{k-2\ell-1}) = P(A_{\ell-1}>B_{\ell-1}),\,\, \ell\in\mathscr{L}_{1k},\nonumber\\
	&P(A'_{k-2\ell-1}>B'_{k-2\ell}) = P(A_{k-1-\ell}>B_{k-1-\ell}),\,\, \ell\in\mathscr{L}_{1k};\nonumber\\
	&P(A'_{k-2\ell-1}<B'_{k-2\ell-2}) = P(A_{k-1-\ell}<B_{k-2-\ell}),\,\, \ell\in\mathscr{L}_{0k-};\nonumber\\
	&P(A'_{0}<B'_{0}) = P(A_{\frac{k-1}{2}}<B_{\frac{k-1}{2}-1}).
\end{align}
In addition, the conditions of~\cref{Extend_Hardy_e} become
\begin{align}\label{Eq:Transformed2}
	&P(A'_{k-2\ell-2}<B'_{k-2\ell-1}) = P(A_{\ell}<B_{\ell-1}),\,\, \ell\in\mathscr{L}_{1k-} \text{ and}\nonumber\\
	&P(A'_{k-2}<B'_{k-1}) = P(A_{0}<B_{k-1}),
\end{align}
To summarize, the requirement that the probabilities in the first two lines of~\cref{Eq:Transformed} vanish is identical to the condition of~\cref{Extend_Stapp_b}, the requirement that the probabilities in the last two lines of~\cref{Eq:Transformed} and the first line of ~\cref{Eq:Transformed2} vanish is identical to the condition of~\cref{Extend_Stapp_c}, and the requirement that the probability in the last line of~\cref{Eq:Transformed2} vanishes is identical to the condition of~\cref{Extend_Stapp_d}.

Similarly, for {\em even} $k$, one can verify that the following relabeling
\begin{equation}
\begin{split}
	A'_{x'}=\left\{
	\begin{array}{ll}
	d-1-A_{\ell-1}, & x' = k-2\ell~\land~ \ell\in\mathscr{L}_{1k}\\
	d-1-A_{k-1-\ell}, & x' = k-1-2\ell~\land~\ell\in\mathscr{L}_{0k}
	\end{array}\right.,\\
	B'_{y'}=\left\{
	\begin{array}{ll}
	d-1-B_{k-1}, & y'=k-1\\
	d-1-B_{k-1-\ell}, & y' = k-2\ell~\land~\ell\in\mathscr{L}_{1k}\\
	d-1-B_{\ell-1}, & y' = k-1-2\ell~\land~ \ell\in\mathscr{L}_{1k}
	\end{array}\right.,\\
\end{split}
\end{equation}
transforms the $2k$ conditions of~\cref{Extend_Hardy} to those of~\cref{Extend_Stapp}.
\end{proof}

For completeness, we show in \cref{Extend_Input} bounds on the maximal DS found for different Hardy-type paradoxes in several $k$-input $2$-output and $k$-input $3$-output Bell scenarios; analogous results for a larger number of outputs but with $k$ set to $2$ are shown in \cref{Extend_Output}. One thing worth noticing is that for all these numerical results, we observe that the DS from our FTI arguments is always higher than that obtained from all these other proposals.

\begin{table}[h!t]
\captionsetup{justification=RaggedRight}
\begin{subtable}[h]{0.5\textwidth}
        \centering
    \begin{tabular}{|c|c||c|c|c|c|c|}
    \hline
    \multicolumn{2}{|c||}{$k$} & $2$ & $3$ & $4$ & $5$ & $6$\\
    \hhline{|==::=|=|=|=|=|}
    \multirow{2}{*}{Boschi~\cite{Boschi97}} & UB & $0.09020$ & $0.17459$ & $0.23132$ & $0.27095$ & $0.29999$ \\
    \hhline{~-|-|-|-|-|-|}
    &LB& $0.09017^*$ & $0.17455$ & $0.23126$ & $0.27088$ & $0.29995$\\
    \hline
    \multirow{2}{*}{Cereceda~\cite{Cereceda16}} & UB & $0.10785$ & $0.18523$ & $0.23801$ & $0.27546$ & $0.30327$ \\
    \hhline{~-|-|-|-|-|-|}
    & LB & $0.10781^*$ & $0.18519$ & $0.23796$ & $0.27542$ & $0.30321$ \\
    \hline
    \multirow{2}{*}{FTI-based} & UB & $0.12501$ & $0.20713$ & $0.25976$ & $0.29579$ & $0.32192$ \\
    \hhline{~-|-|-|-|-|-|}
    & LB & $0.125^*$ & $0.20711$ & $0.25973$ & $0.29576$ & $0.32190$ \\
    \hline
\end{tabular}
\caption{Results for binary-output Bell scenarios. From top to bottom, we list the results for Hardy paradox given by the ladder proof of nonlocality without inequality~\cite{Boschi97} (see also~\cite{Meng:PRA:2018}), generalized CLL's argument due to~\cite{Cereceda16}, and our FTI-based formulation~[\cref{Extend_Stapp}].
The quantum strategies corresponding to the LBs are provided in~\cite{Supplemental}.
 Entries marked with * are known to be tight quantum bound.}\label{Extend_Input_2Output}
\end{subtable} 
\begin{subtable}[h]{0.5\textwidth}
        \centering
    \begin{tabular}{|c|c||c|c|c|c|}
    \hline
    \multicolumn{2}{|c||}{$k$} & $2$ & $3$ & $4$ & $5$\\
    \hhline{|==::=|=|=|=|}
    \multirow{2}{*}{Meng \cite{Meng:PRA:2018}}& UB & $0.14194$ & $0.26782$ & $0.34823$ & $0.40196$  \\
    \hhline{~-|-|-|-|-|}
    & LB & $ 0.14133$ & $0.26779$ & $0.34816$ & $0.40184$  \\
    \hline
    \multirow{2}{*}{CLL-type} & UB & $0.16794$ & $0.28272$ & $0.35706$ & $0.40763$  \\
    \hhline{~-|-|-|-|-|}
    & LB & $0.16791$ & $0.28265$ & $0.35698$ & $0.40753$  \\
    \hline
    \multirow{2}{*}{FTI-based} & UB & $0.19313$ & $0.31230$ & $0.38467$ & $0.43225$  \\
    \hhline{~-|-|-|-|-|}
    & LB & $0.19309$ & $0.31226$ & $0.38460$ & $0.43216$  \\
    \hline
\end{tabular}
\caption{Results for ternary-output Bell scenarios. From top to bottom, we list the results for generalized Hardy paradox given by Meng {\em et al.}~\cite{Meng:PRA:2018}, our generalized CLL formulation [\cref{Extend_Hardy}], and our FTI-based formulation [\cref{Extend_Stapp}]. The quantum strategies corresponding to the LBs are provided in~\cite{Supplemental}.
}\label{Extend_Input_3Output}
 \end{subtable}
    \caption{Comparison of the best upper bound (UB) and the best lower bound (LB) found on the DS for three different Hardy and Hardy-type paradoxes beyond the CHSH scenarios with $k\ge 2$ inputs. See also~\cref{Fig:Summary} for the relationship between all these paradoxes.
    The UBs were obtained by considering level-1\footnote{\label{level} Not all upper bounds reported here were obtained using the higher-level SDP hierarchy because some of these higher-level computations, due to numerical issues, resulted in worse upper bounds.} of the SDP hierarchy introduced in Ref.~\cite{Moroder13}. The LBs, obtained numerically, may be attained using the strategies provided in~\cite{Supplemental} (see also~\cref{App:Opt_Q}).
    }
    \label{Extend_Input}
\end{table}

\begin{table}
\captionsetup{justification=RaggedRight}
\centering
\scalebox{0.9}{\begin{tabular}{|c|c||c|c|c|c|c|c|}
    \hline
    \multicolumn{2}{|c||}{$d$} & $2$ & $3$ & $4$ & $5$ & $6$ &$7$\\
    \hhline{|==::=|=|=|=|=|=|}
    \multirow{2}{*}{Chen~\cite{Chen:PRA:2013}} & UB & $0.09020$ & $0.14194$ & $ 0.17659^\ddag$ & $0.20317^\ddag$ & $ 0.22441^\ddag$ & $ 0.24196^\ddag$\\
    \hhline{~-|-|-|-|-|-|-|}
    & LB & $0.09017^*$ & $0.14133$ & $ 0.17656$ & $ 0.20306$ & $ 0.22424$ & $0.24175$\\
    \hline
    \multirow{2}{*}{CLL-type} & UB  & $0.10785$ & $0.16794$ & $0.20890^\ddag$ & $ 0.23959^\ddag$ & $0.26392^\ddag$ & $0.28395^\ddag$\\
    \hhline{~-|-|-|-|-|-|-|}
    & LB  & $0.10781^*$ & $0.16791$ & $0.20883$ & $0.23948$ & $0.26378$ & $0.28378$\\
    \hline
    \multirow{2}{*}{FTI-based} & UB & $0.12501$ & $0.19313$ & $ 0.23844$ & $0.27176^\ddag$ & $0.29782^\ddag$ & $0.31904$\\
    \hhline{~-|-|-|-|-|-|-|}
    & LB & $0.125^*$ & $0.19309$ & $0.23839$ & $0.27175$ & $0.29773$ & $0.31880$\\
    \hline
\end{tabular}}
    \caption{Comparison of the best upper bound (UB) and the best lower bound (LB) found on the DS for the Hardy paradox due to~\cite{Chen:PRA:2013}, our generalized CLL formulation [\cref{Extend_Hardy}], and our FTI-based formulation [\cref{Extend_Stapp}] for Bell scenarios with two inputs and $d\ge 2$ outputs. See also~\cref{Fig:Summary} for the relationship between all these paradoxes.
    Most of the UBs were obtained by considering level-1 of the SDP hierarchy introduced in Ref.~\cite{Moroder13} but those marked with $^\ddag$ were obtained by considering level-2 of the SDP hierarchy introduced in Ref.~\cite{NPA,NPA2008}} (see also~\cref{level}). The LBs, obtained numerically, may be attained using the strategies provided in \cite{Supplemental} (see also~\cref{App:Opt_Q}). Entries marked with * are known to be tight quantum bound.
    
    \label{Extend_Output}
\end{table}

\section{Discussion}

Hardy and Hardy-type paradoxes are fascinating proofs of Bell nonlocality without resorting to Bell inequalities. Aside from fundamental interests (see, e.g.,~\cite{PRB:PRL:2014,Goh2018,Rai:PRA:2019,Chen23}), they are also known to be relevant in the task of randomness amplification~\cite{CR:NatPhys:2012} (see, e.g.,~\cite{Ravi18,ZRL+23}). In this work, we propose a Hardy-type paradox that can be naturally understood via the failure of the transitivity of implications (FTI), cf.~\cite{Stapp93,Liang:PRep}. 

As with the Hardy-type paradoxes formulated by Cabello-Liang-Li (CLL)~\cite{Cabello02,Liang:2005vl}, we show that a degree of success (DS) generalizing the notion of success probability---whose non-negative value witnesses Bell-nonlocality---may be introduced for the FTI-based Hardy-type paradox. In the simplest Bell scenario with two inputs and two outputs, we show that the new FTI-based formulations give the highest DS among all existing (i.e., Hardy, CLL, and FTI-based) formulations. Moreover, this advantage---in the form of a generalized DS---persists even when the zero-probability constraints required in all these formulations are relaxed. 

Then, we provide---as with~\cite{Meng:PRA:2018} for the original Hardy paradox---a generalization of the FTI-based formulation and the CLL-type formulation, to symmetric Bell scenarios involving an {\em arbitrary} number of inputs and outputs. In turn, this allows us to show that a ladder-type, cf.~\cite{Boschi97}, and an FTI-based proof of nonlocality without inequality are equivalent for an arbitrary symmetric Bell scenario, thereby generalizing the result of~\cite{Meng:PRA:2018} beyond the binary-outcome Bell scenarios.

For several simple Bell scenarios, we further observe (see \cref{Extend_Input} and \cref{Extend_Output}) numerically that our FTI-based generalizations provide the largest value of the DS. We do not currently have any concrete physical intuition behind this observation but it will be interesting to develop one in the future. Another natural question left open from the current work is to determine if this trend continues to hold for an arbitrary, symmetric Bell scenario. Also worth noting is that within each type of logical argument, the largest values of DS found appear to increase monotonically when one increases either the number of inputs $k$ or the number of outputs $d$ involved  --- an analytic proof of this observation will be more than welcome. 

On the other hand, given the close connection found~\cite{Rai:PRA:2022,Chen23,Liu:2309.06304} between the optimizing strategy for a Hardy paradox and its self-testing~\cite{Supic19} property, it would also be interesting to see if the optimizing correlations found for these new generalizations are also self-testing (and non-exposed~\cite{Goh2018}). From an application perspective, one may also be interested in the potential of such correlations for device-independent applications, especially in randomness amplification~\cite{CR:NatPhys:2012}, and proofs of Bell-nonlocality in the presence of measurement independence~\cite{PRB:PRL:2014}.

\begin{acknowledgements}
We thank Ashutosh Rai for helpful discussions and two anonymous reviewers for very helpful suggestions. This work is supported by the National Science and Technology Council, Taiwan (Grants No. 109-2112-M006-010-MY3, 112-2628-M006-007-MY4).
\end{acknowledgements}

\appendix

\section{Degree of Success versus Concurrence }\label{App:DSvsC}

For Hardy's argument, we can again take~\cref{Eq:Class2b} but now with the constraint~\cite{Chen23} 
\begin{equation}\label{Eq:HardyParametricConstraint}
	\tan\theta\sin\alpha = \tan\beta.
\end{equation}	
Hence, we again have the concurrence given by~\cref{Eq:Concurrence_Constraint}. Moreover, the success probability of \cref{HardyParadox} is 
\begin{equation}
    P_{\text{succ}}(\theta,\alpha,\beta)=(\cos\theta\cos\alpha\sin\beta)^2,
\end{equation}
subjected to \cref{Eq:HardyParametricConstraint}.
Rewriting $P_{\text{succ}}$ in terms of $C$ and $\beta$ and using variational techniques, the largest DS is obtained for $\cos^2\beta = \frac{C}{2-C}$, giving
\begin{equation}\label{Eq:Psucc_FixedTheta}
    P^{*}_{\text{succ}}(C)=\frac{C^2(1-C)}{(2-C)^2}.
\end{equation}

Similarly, for the CLL argument using the state and observables given in~\cite[Eqs. (40, 41)]{Chen23}, we have
\begin{equation}
    C(\ket{\psi}) = |\cos^2\phi \sin{2\theta}|
\end{equation}
and the degree of success
\begin{align}
    P_{\text{succ}}(\phi,\theta,\alpha,\beta) &=(\cos\phi \sin\theta\cos\alpha)^2  \nonumber \\
     &- (\sin\phi\cos\beta + \cos\phi\sin\theta\sin\beta)^2,
\end{align}
with constraint $(\tan\phi + \cos\theta\tan\alpha)\tan\beta = \sin\theta$.
To obtain the maximal DS for a given concurrence $C$, we may take $\alpha = \frac{\pi}{2} - \beta$ and use the constraint to eliminate $\phi$ and $\theta$, thus arriving at
\begin{align}\label{Eq:DS:C}
    P_{\text{succ}}(C,\beta) &= 
    (C-1)\cos^4{\beta} + \sqrt{2}\cos\beta\sin\beta \nonumber \\
    &\times \sqrt{(1-C)\cos^2{\beta}\left[1+C+(C-1)\cos{2\beta}\right]} \nonumber \\
    &\times \sin\left[\frac{1}{2}\sin^{-1}\left(\frac{C \sec^2{\beta}}{C+\tan^2{\beta}}\right)\right].
\end{align}
Using, e.g., Mathematica, we can numerically optimize $P_{\text{succ}}(C,\beta)$ for fixed values of $C$ and verify that the resulting plot matches the green dashed line in Fig.~\ref{fig:Compare_state}.

\section{Quantum Strategies}\label{App:Opt_Q}

In this Appendix, we give some further information about the quantum strategies that reproduce our best lower bound (LB) on the various DS shown in \cref{Extend_Input} and \cref{Extend_Output}. The actual quantum strategies for each case are available at~\cite{Supplemental}. For convenience, we refer to the various generalizations of Hardy~\cite{Hardy93} due to Boschi {\em et al.}~\cite{Boschi97}, Chen {\em et al.}~\cite{Chen:PRA:2013}, and Meng {\em et al.}~\cite{Meng:PRA:2018} as Hardy paradoxes. Indeed, in all these three cases, the DS is exactly the success probability of observing such a paradox.

Next, notice that the best quantum strategies we have found for the Hardy, CLL-type, and FTI-based arguments for a $k$-input, $d$-output Bell scenario, denoted by $(k,d)$, are always attained by performing real rank-1 projective measurements on a real pure quantum state $\ket{\psi}$ residing in the two-qudit Hilbert space $\mathbb{R}^d\otimes\mathbb{R}^d$. However, due to numerical imprecisions, the zero-probability constraints are {\em not} always strictly enforced in all these optimal strategies found. For completeness, we list in \cref{Tbl:Error} the largest deviation found among all the zero-probability constraints for each of these ``optimal" strategies. 

\begin{table}[h!]
\captionsetup{justification=RaggedRight}
  \begin{tabular}{|c|c|c|c|c|}\hline
    $(k,d)$  &  Hardy & CLL & FTI \\ \hline \hline 
	$(3, 2)$  &  $2.4057 \times 10^{-14}$ & 	$1.1211 \times 10^{-14}$		&  $9.7925 \times 10^{-16}$ \\ \hline
	$(4, 2)$ & $3.7182 \times 10^{-14}$ & 	$1.5057 \times 10^{-15}$		&  $7.1356 \times 10^{-16}$ \\ \hline
	$(5, 2)$  & $1.2657 \times 10^{-15}$ & 	$3.2069 \times 10^{-10}$		&  $1.5557 \times 10^{-12}$ \\ \hline
	$(6, 2)$  & $2.2401 \times 10^{-15}$ & 	$1.6756 \times 10^{-11}$		&  $4.0152 \times 10^{-12}$ \\ \hline \hline	
	$(2, 3)$  & $2.8391 \times 10^{-15}$ & 	$1.2214 \times 10^{-16}$		&  $2.9997 \times 10^{-16}$ \\ \hline
	$(3, 3)$  & $2.4822 \times 10^{-9}$ & 	$4.9471 \times 10^{-16}$		&  $3.8448 \times 10^{-10}$ \\ \hline
	$(4, 3)$  & $1.6744 \times 10^{-11}$ & 	$2.0745 \times 10^{-14}$		&  $4.8225 \times 10^{-14}$ \\ \hline
	$(5, 3)$  & $9.9972 \times 10^{-14}$ & 	$4.7729 \times 10^{-14}$		&  $4.8075 \times 10^{-12}$ \\ \hline \hline
	$(2, 4)$  & $8.0796 \times 10^{-8}$ & 	$1.8599 \times 10^{-16}$		&  $3.2093 \times 10^{-16}$ \\ \hline
	$(2, 5)$  & $3.5385 \times 10^{-10}$ & 	$2.3190 \times 10^{-14}$		&  $5.0784 \times 10^{-8}$ \\ \hline
	$(2, 6)$  & $4.3552 \times 10^{-10}$ & 	$4.7699 \times 10^{-11}$	&  $7.8787 \times 10^{-11}$ \\ \hline
	$(2, 7)$  &  $1.2624 \times 10^{-7}$ & 	$ 2.3396 \times 10^{-12}$		&  $5.9839 \times 10^{-8}$ \\ \hline
  \end{tabular}
   \caption{\label{Tbl:Error} Summary of the largest zero-probability-constraint violation for each optimal strategy~\cite{Supplemental} used to give the best lower bound on the DS presented in~\cref{Extend_Input,Extend_Output}.}
\end{table}

\end{CJK*}


\begin{thebibliography}{40}%
\makeatletter
\providecommand \@ifxundefined [1]{%
 \@ifx{#1\undefined}
}%
\providecommand \@ifnum [1]{%
 \ifnum #1\expandafter \@firstoftwo
 \else \expandafter \@secondoftwo
 \fi
}%
\providecommand \@ifx [1]{%
 \ifx #1\expandafter \@firstoftwo
 \else \expandafter \@secondoftwo
 \fi
}%
\providecommand \natexlab [1]{#1}%
\providecommand \enquote  [1]{``#1''}%
\providecommand \bibnamefont  [1]{#1}%
\providecommand \bibfnamefont [1]{#1}%
\providecommand \citenamefont [1]{#1}%
\providecommand \href@noop [0]{\@secondoftwo}%
\providecommand \href [0]{\begingroup \@sanitize@url \@href}%
\providecommand \@href[1]{\@@startlink{#1}\@@href}%
\providecommand \@@href[1]{\endgroup#1\@@endlink}%
\providecommand \@sanitize@url [0]{\catcode `\\12\catcode `\$12\catcode
  `\&12\catcode `\#12\catcode `\^12\catcode `\_12\catcode `\%12\relax}%
\providecommand \@@startlink[1]{}%
\providecommand \@@endlink[0]{}%
\providecommand \url  [0]{\begingroup\@sanitize@url \@url }%
\providecommand \@url [1]{\endgroup\@href {#1}{\urlprefix }}%
\providecommand \urlprefix  [0]{URL }%
\providecommand \Eprint [0]{\href }%
\providecommand \doibase [0]{http://dx.doi.org/}%
\providecommand \selectlanguage [0]{\@gobble}%
\providecommand \bibinfo  [0]{\@secondoftwo}%
\providecommand \bibfield  [0]{\@secondoftwo}%
\providecommand \translation [1]{[#1]}%
\providecommand \BibitemOpen [0]{}%
\providecommand \bibitemStop [0]{}%
\providecommand \bibitemNoStop [0]{.\EOS\space}%
\providecommand \EOS [0]{\spacefactor3000\relax}%
\providecommand \BibitemShut  [1]{\csname bibitem#1\endcsname}%
\let\auto@bib@innerbib\@empty
\bibitem [{\citenamefont {Einstein}\ \emph {et~al.}(1935)\citenamefont
  {Einstein}, \citenamefont {Podolsky},\ and\ \citenamefont {Rosen}}]{EPR35}%
  \BibitemOpen
  \bibfield  {author} {\bibinfo {author} {\bibfnamefont {A.}~\bibnamefont
  {Einstein}}, \bibinfo {author} {\bibfnamefont {B.}~\bibnamefont {Podolsky}},
  \ and\ \bibinfo {author} {\bibfnamefont {N.}~\bibnamefont {Rosen}},\ }\href
  {\doibase 10.1103/PhysRev.47.777} {\bibfield  {journal} {\bibinfo  {journal}
  {Phys. Rev.}\ }\textbf {\bibinfo {volume} {47}},\ \bibinfo {pages} {777}
  (\bibinfo {year} {1935})}\BibitemShut {NoStop}%
\bibitem [{\citenamefont {Bell}(1964)}]{Bell64}%
  \BibitemOpen
  \bibfield  {author} {\bibinfo {author} {\bibfnamefont {J.~S.}\ \bibnamefont
  {Bell}},\ }\href {\doibase 10.1103/PhysicsPhysiqueFizika.1.195} {\bibfield
  {journal} {\bibinfo  {journal} {Physics}\ }\textbf {\bibinfo {volume} {1}},\
  \bibinfo {pages} {195} (\bibinfo {year} {1964})}\BibitemShut {NoStop}%
\bibitem [{\citenamefont {Brunner}\ \emph {et~al.}(2014)\citenamefont
  {Brunner}, \citenamefont {Cavalcanti}, \citenamefont {Pironio}, \citenamefont
  {Scarani},\ and\ \citenamefont {Wehner}}]{Brunner_RevModPhys_2014}%
  \BibitemOpen
  \bibfield  {author} {\bibinfo {author} {\bibfnamefont {N.}~\bibnamefont
  {Brunner}}, \bibinfo {author} {\bibfnamefont {D.}~\bibnamefont {Cavalcanti}},
  \bibinfo {author} {\bibfnamefont {S.}~\bibnamefont {Pironio}}, \bibinfo
  {author} {\bibfnamefont {V.}~\bibnamefont {Scarani}}, \ and\ \bibinfo
  {author} {\bibfnamefont {S.}~\bibnamefont {Wehner}},\ }\href {\doibase 10.1103/RevModPhys.86.419} {\bibfield  {journal} {\bibinfo
  {journal} {Rev. Mod. Phys.}\ }\textbf {\bibinfo {volume} {86}},\ \bibinfo
  {pages} {419} (\bibinfo {year} {2014})}\BibitemShut {NoStop}%
\bibitem [{\citenamefont {Scarani}(2012)}]{Scarani_DIQI_12}%
  \BibitemOpen
  \bibfield  {author} {\bibinfo {author} {\bibfnamefont {V.}~\bibnamefont
  {Scarani}},\ }\href@noop {} {\bibfield  {journal} {\bibinfo  {journal} {Acta
  Physica Slovaca}\ }\textbf {\bibinfo {volume} {62}},\ \bibinfo {pages} {347}
  (\bibinfo {year} {2012})}\BibitemShut {NoStop}%
\bibitem [{\citenamefont {Greenberger}\ \emph {et~al.}(1989)\citenamefont
  {Greenberger}, \citenamefont {Horne},\ and\ \citenamefont
  {Zeilinger}}]{Greenberger:1989aa}%
  \BibitemOpen
  \bibfield  {author} {\bibinfo {author} {\bibfnamefont {D.~M.}\ \bibnamefont
  {Greenberger}}, \bibinfo {author} {\bibfnamefont {M.~A.}\ \bibnamefont
  {Horne}}, \ and\ \bibinfo {author} {\bibfnamefont {A.}~\bibnamefont
  {Zeilinger}},\ }\enquote {\bibinfo {title} {Going beyond bell's theorem},}\
  in\ \href {\doibase 10.1007/978-94-017-0849-4_10} {\emph {\bibinfo
  {booktitle} {Bell's Theorem, Quantum Theory and Conceptions of the
  Universe}}},\ \bibinfo {editor} {edited by\ \bibinfo {editor} {\bibfnamefont
  {M.}~\bibnamefont {Kafatos}}}\ (\bibinfo  {publisher} {Springer
  Netherlands},\ \bibinfo {address} {Dordrecht},\ \bibinfo {year} {1989})\ pp.\
  \bibinfo {pages} {69--72}\BibitemShut {NoStop}%
\bibitem [{\citenamefont {Greenberger}\ \emph {et~al.}(1990)\citenamefont
  {Greenberger}, \citenamefont {Horne}, \citenamefont {Shimony},\ and\
  \citenamefont {Zeilinger}}]{GHSZ}%
  \BibitemOpen
  \bibfield  {author} {\bibinfo {author} {\bibfnamefont {D.~M.}\ \bibnamefont
  {Greenberger}}, \bibinfo {author} {\bibfnamefont {M.~A.}\ \bibnamefont
  {Horne}}, \bibinfo {author} {\bibfnamefont {A.}~\bibnamefont {Shimony}}, \
  and\ \bibinfo {author} {\bibfnamefont {A.}~\bibnamefont {Zeilinger}},\ }\href
  {\doibase 10.1119/1.16243} {\bibfield  {journal} {\bibinfo
  {journal} {Am. J. Phys.}\ }\textbf {\bibinfo {volume} {58}},\ \bibinfo
  {pages} {1131} (\bibinfo {year} {1990})}\BibitemShut {NoStop}%
\bibitem [{\citenamefont {Peres}(1990)}]{Peres:1990ab}%
  \BibitemOpen
  \bibfield  {author} {\bibinfo {author} {\bibfnamefont {A.}~\bibnamefont
  {Peres}},\ }\href {\doibase10.1016/0375-9601(90)90172-K}
  {\bibfield  {journal} {\bibinfo  {journal} {Phys. Lett. A}\ }\textbf
  {\bibinfo {volume} {151}},\ \bibinfo {pages} {107} (\bibinfo {year}
  {1990})}\BibitemShut {NoStop}%
\bibitem [{\citenamefont {Mermin}(1990)}]{Mermin_90}%
  \BibitemOpen
  \bibfield  {author} {\bibinfo {author} {\bibfnamefont {N.~D.}\ \bibnamefont
  {Mermin}},\ }\href {\doibase 10.1103/PhysRevLett.65.3373} {\bibfield
  {journal} {\bibinfo  {journal} {Phys. Rev. Lett.}\ }\textbf {\bibinfo
  {volume} {65}},\ \bibinfo {pages} {3373} (\bibinfo {year}
  {1990})}\BibitemShut {NoStop}%
\bibitem [{\citenamefont {Hardy}(1992)}]{Hardy92}%
  \BibitemOpen
  \bibfield  {author} {\bibinfo {author} {\bibfnamefont {L.}~\bibnamefont
  {Hardy}},\ }\href {\doibase 10.1103/PhysRevLett.68.2981} {\bibfield
  {journal} {\bibinfo  {journal} {Phys. Rev. Lett.}\ }\textbf {\bibinfo
  {volume} {68}},\ \bibinfo {pages} {2981} (\bibinfo {year}
  {1992})}\BibitemShut {NoStop}%
\bibitem [{\citenamefont {Clifton}\ and\ \citenamefont
  {Niemann}(1992)}]{Clifton92}%
  \BibitemOpen
  \bibfield  {author} {\bibinfo {author} {\bibfnamefont {R.}~\bibnamefont
  {Clifton}}\ and\ \bibinfo {author} {\bibfnamefont {P.}~\bibnamefont
  {Niemann}},\ }\href {\doibase 10.1016/0375-9601(92)90358-S}
  {\bibfield  {journal} {\bibinfo  {journal} {Phys. Lett. A}\ }\textbf
  {\bibinfo {volume} {166}},\ \bibinfo {pages} {177} (\bibinfo {year}
  {1992})}\BibitemShut {NoStop}%
\bibitem [{\citenamefont {Pagonis}\ and\ \citenamefont
  {Clifton}(1992)}]{Pagonis92}%
  \BibitemOpen
  \bibfield  {author} {\bibinfo {author} {\bibfnamefont {C.}~\bibnamefont
  {Pagonis}}\ and\ \bibinfo {author} {\bibfnamefont {R.}~\bibnamefont
  {Clifton}},\ }\href {\doibase 10.1016/0375-9601(92)90070-3}
  {\bibfield  {journal} {\bibinfo  {journal} {Phys. Lett. A}\ }\textbf
  {\bibinfo {volume} {168}},\ \bibinfo {pages} {100} (\bibinfo {year}
  {1992})}\BibitemShut {NoStop}%
\bibitem [{\citenamefont {Jiang}\ \emph {et~al.}(2018)\citenamefont {Jiang},
  \citenamefont {Xu}, \citenamefont {Su}, \citenamefont {Pati},\ and\
  \citenamefont {Chen}}]{Jiang:PRL:2018}%
  \BibitemOpen
  \bibfield  {author} {\bibinfo {author} {\bibfnamefont {S.-H.}\ \bibnamefont
  {Jiang}}, \bibinfo {author} {\bibfnamefont {Z.-P.}\ \bibnamefont {Xu}},
  \bibinfo {author} {\bibfnamefont {H.-Y.}\ \bibnamefont {Su}}, \bibinfo
  {author} {\bibfnamefont {A.~K.}\ \bibnamefont {Pati}}, \ and\ \bibinfo
  {author} {\bibfnamefont {J.-L.}\ \bibnamefont {Chen}},\ }\href {\doibase 10.1103/PhysRevLett.120.050403} {\bibfield  {journal}
  {\bibinfo  {journal} {Phys. Rev. Lett.}\ }\textbf {\bibinfo {volume} {120}},\
  \bibinfo {pages} {050403} (\bibinfo {year} {2018})}\BibitemShut {NoStop}%
\bibitem [{\citenamefont {Luo}\ \emph {et~al.}(2018)\citenamefont {Luo},
  \citenamefont {Su}, \citenamefont {Huang}, \citenamefont {Wang},
  \citenamefont {Yang}, \citenamefont {Li}, \citenamefont {Liu}, \citenamefont
  {Chen}, \citenamefont {Lu},\ and\ \citenamefont {Pan}}]{Luo:2018aa}%
  \BibitemOpen
  \bibfield  {author} {\bibinfo {author} {\bibfnamefont {Y.-H.}\ \bibnamefont
  {Luo}}, \bibinfo {author} {\bibfnamefont {H.-Y.}\ \bibnamefont {Su}},
  \bibinfo {author} {\bibfnamefont {H.-L.}\ \bibnamefont {Huang}}, \bibinfo
  {author} {\bibfnamefont {X.-L.}\ \bibnamefont {Wang}}, \bibinfo {author}
  {\bibfnamefont {T.}~\bibnamefont {Yang}}, \bibinfo {author} {\bibfnamefont
  {L.}~\bibnamefont {Li}}, \bibinfo {author} {\bibfnamefont {N.-L.}\
  \bibnamefont {Liu}}, \bibinfo {author} {\bibfnamefont {J.-L.}\ \bibnamefont
  {Chen}}, \bibinfo {author} {\bibfnamefont {C.-Y.}\ \bibnamefont {Lu}}, \ and\
  \bibinfo {author} {\bibfnamefont {J.-W.}\ \bibnamefont {Pan}},\ }\href
  {\doibase 10.1016/j.scib.2018.11.025} {\bibfield  {journal}
  {\bibinfo  {journal} {Sci. Bull.}\ }\textbf {\bibinfo {volume} {63}},\
  \bibinfo {pages} {1611} (\bibinfo {year} {2018})}\BibitemShut {NoStop}%
\bibitem [{\citenamefont {Hardy}(1993)}]{Hardy93}%
  \BibitemOpen
  \bibfield  {author} {\bibinfo {author} {\bibfnamefont {L.}~\bibnamefont
  {Hardy}},\ }\href {\doibase 10.1103/PhysRevLett.71.1665} {\bibfield
  {journal} {\bibinfo  {journal} {Phys. Rev. Lett.}\ }\textbf {\bibinfo
  {volume} {71}},\ \bibinfo {pages} {1665} (\bibinfo {year}
  {1993})}\BibitemShut {NoStop}%
\bibitem [{\citenamefont {Boschi}\ \emph {et~al.}(1997)\citenamefont {Boschi},
  \citenamefont {Branca}, \citenamefont {De~Martini},\ and\ \citenamefont
  {Hardy}}]{Boschi97}%
  \BibitemOpen
  \bibfield  {author} {\bibinfo {author} {\bibfnamefont {D.}~\bibnamefont
  {Boschi}}, \bibinfo {author} {\bibfnamefont {S.}~\bibnamefont {Branca}},
  \bibinfo {author} {\bibfnamefont {F.}~\bibnamefont {De~Martini}}, \ and\
  \bibinfo {author} {\bibfnamefont {L.}~\bibnamefont {Hardy}},\ }\href
  {\doibase 10.1103/PhysRevLett.79.2755} {\bibfield  {journal} {\bibinfo
  {journal} {Phys. Rev. Lett.}\ }\textbf {\bibinfo {volume} {79}},\ \bibinfo
  {pages} {2755} (\bibinfo {year} {1997})}\BibitemShut {NoStop}%
\bibitem [{\citenamefont {Stapp}(1993)}]{Stapp93}%
  \BibitemOpen
  \bibfield  {author} {\bibinfo {author} {\bibfnamefont {H.~P.}\ \bibnamefont
  {Stapp}},\ }\href@noop {} {\emph {\bibinfo {title} {Mind, Matter and Quantum
  Mechanics}}}\ (\bibinfo  {publisher} {Springer Verlag},\ \bibinfo {year}
  {1993})\BibitemShut {NoStop}%
\bibitem [{\citenamefont {Liang}\ \emph {et~al.}(2011)\citenamefont {Liang},
  \citenamefont {Spekkens},\ and\ \citenamefont {Wiseman}}]{Liang:PRep}%
  \BibitemOpen
  \bibfield  {author} {\bibinfo {author} {\bibfnamefont {Y.-C.}\ \bibnamefont
  {Liang}}, \bibinfo {author} {\bibfnamefont {R.~W.}\ \bibnamefont {Spekkens}},
  \ and\ \bibinfo {author} {\bibfnamefont {H.~M.}\ \bibnamefont {Wiseman}},\
  }\href {\doibase 10.1016/j.physrep.2011.05.001} {\bibfield  {journal}
  {\bibinfo  {journal} {Phys. Rep.}\ }\textbf {\bibinfo {volume} {506}},\
  \bibinfo {pages} {1 } (\bibinfo {year} {2011})}\BibitemShut {NoStop}%
\bibitem [{\citenamefont {Kar}(1997)}]{Kar97}%
  \BibitemOpen
  \bibfield  {author} {\bibinfo {author} {\bibfnamefont {G.}~\bibnamefont
  {Kar}},\ }\href {\doibase 10.1016/S0375-9601(97)00116-3} {\bibfield
  {journal} {\bibinfo  {journal} {Phys. Lett. A}\ }\textbf {\bibinfo {volume}
  {228}},\ \bibinfo {pages} {119} (\bibinfo {year} {1997})}\BibitemShut
  {NoStop}%
\bibitem [{\citenamefont {Liang}\ and\ \citenamefont {Li}(2003)}]{Liang03}%
  \BibitemOpen
  \bibfield  {author} {\bibinfo {author} {\bibfnamefont {L.-M.}\ \bibnamefont
  {Liang}}\ and\ \bibinfo {author} {\bibfnamefont {C.-Z.}\ \bibnamefont {Li}},\
  }\href {\doibase 10.1016/j.physleta.2003.09.038} {\bibfield
  {journal} {\bibinfo  {journal} {Phys. Lett. A}\ }\textbf {\bibinfo {volume}
  {318}},\ \bibinfo {pages} {300} (\bibinfo {year} {2003})}\BibitemShut
  {NoStop}%
\bibitem [{\citenamefont {Liang}\ and\ \citenamefont
  {Li}(2005)}]{Liang:2005vl}%
  \BibitemOpen
  \bibfield  {author} {\bibinfo {author} {\bibfnamefont {L.-M.}\ \bibnamefont
  {Liang}}\ and\ \bibinfo {author} {\bibfnamefont {C.-Z.}\ \bibnamefont {Li}},\
  }\href {\doibase 10.1016/j.physleta.2004.12.046} {\bibfield
  {journal} {\bibinfo  {journal} {Phys. Lett. A}\ }\textbf {\bibinfo {volume}
  {335}},\ \bibinfo {pages} {371} (\bibinfo {year} {2005})}\BibitemShut
  {NoStop}%
\bibitem [{\citenamefont {Cabello}(2002)}]{Cabello02}%
  \BibitemOpen
  \bibfield  {author} {\bibinfo {author} {\bibfnamefont {A.}~\bibnamefont
  {Cabello}},\ }\href {\doibase 10.1103/PhysRevA.65.032108} {\bibfield
  {journal} {\bibinfo  {journal} {Phys. Rev. A}\ }\textbf {\bibinfo {volume}
  {65}},\ \bibinfo {pages} {032108} (\bibinfo {year} {2002})}\BibitemShut
  {NoStop}%
\bibitem [{\citenamefont {Rai}\ \emph {et~al.}(2021)\citenamefont {Rai},
  \citenamefont {Pivoluska}, \citenamefont {Plesch}, \citenamefont {Sasmal},
  \citenamefont {Banik},\ and\ \citenamefont {Ghosh}}]{Rai:PRA:2021}%
  \BibitemOpen
  \bibfield  {author} {\bibinfo {author} {\bibfnamefont {A.}~\bibnamefont
  {Rai}}, \bibinfo {author} {\bibfnamefont {M.}~\bibnamefont {Pivoluska}},
  \bibinfo {author} {\bibfnamefont {M.}~\bibnamefont {Plesch}}, \bibinfo
  {author} {\bibfnamefont {S.}~\bibnamefont {Sasmal}}, \bibinfo {author}
  {\bibfnamefont {M.}~\bibnamefont {Banik}}, \ and\ \bibinfo {author}
  {\bibfnamefont {S.}~\bibnamefont {Ghosh}},\ }\href {\doibase 10.1103/PhysRevA.103.062219} {\bibfield  {journal} {\bibinfo
  {journal} {Phys. Rev. A}\ }\textbf {\bibinfo {volume} {103}},\ \bibinfo
  {pages} {062219} (\bibinfo {year} {2021})}\BibitemShut {NoStop}%
\bibitem [{\citenamefont {Kunkri}\ \emph {et~al.}(2006)\citenamefont {Kunkri},
  \citenamefont {Choudhary}, \citenamefont {Ahanj},\ and\ \citenamefont
  {Joag}}]{KCA+06}%
  \BibitemOpen
  \bibfield  {author} {\bibinfo {author} {\bibfnamefont {S.}~\bibnamefont
  {Kunkri}}, \bibinfo {author} {\bibfnamefont {S.~K.}\ \bibnamefont
  {Choudhary}}, \bibinfo {author} {\bibfnamefont {A.}~\bibnamefont {Ahanj}}, \
  and\ \bibinfo {author} {\bibfnamefont {P.}~\bibnamefont {Joag}},\ }\href
  {\doibase 10.1103/PhysRevA.73.022346} {\bibfield  {journal} {\bibinfo
  {journal} {Phys. Rev. A}\ }\textbf {\bibinfo {volume} {73}},\ \bibinfo
  {pages} {022346} (\bibinfo {year} {2006})}\BibitemShut {NoStop}%
\bibitem [{\citenamefont {Cereceda}(2017)}]{Cereceda16}%
  \BibitemOpen
  \bibfield  {author} {\bibinfo {author} {\bibfnamefont {J.~L.}\ \bibnamefont
  {Cereceda}},\ }\href {\doibase 10.1007/s40509-016-0093-7} {\bibfield
  {journal} {\bibinfo  {journal} {Quantum Studies: Mathematics and
  Foundations}\ }\textbf {\bibinfo {volume} {4}},\ \bibinfo {pages} {205}
  (\bibinfo {year} {2017})}\BibitemShut {NoStop}%
\bibitem [{\citenamefont {Chen}\ \emph {et~al.}(2013)\citenamefont {Chen},
  \citenamefont {Cabello}, \citenamefont {Xu}, \citenamefont {Su},
  \citenamefont {Wu},\ and\ \citenamefont {Kwek}}]{Chen:PRA:2013}%
  \BibitemOpen
  \bibfield  {author} {\bibinfo {author} {\bibfnamefont {J.-L.}\ \bibnamefont
  {Chen}}, \bibinfo {author} {\bibfnamefont {A.}~\bibnamefont {Cabello}},
  \bibinfo {author} {\bibfnamefont {Z.-P.}\ \bibnamefont {Xu}}, \bibinfo
  {author} {\bibfnamefont {H.-Y.}\ \bibnamefont {Su}}, \bibinfo {author}
  {\bibfnamefont {C.}~\bibnamefont {Wu}}, \ and\ \bibinfo {author}
  {\bibfnamefont {L.~C.}\ \bibnamefont {Kwek}},\ }\href {\doibase 10.1103/PhysRevA.88.062116} {\bibfield  {journal} {\bibinfo
  {journal} {Phys. Rev. A}\ }\textbf {\bibinfo {volume} {88}},\ \bibinfo
  {pages} {062116} (\bibinfo {year} {2013})}\BibitemShut {NoStop}%
\bibitem [{\citenamefont {Meng}\ \emph {et~al.}(2018)\citenamefont {Meng},
  \citenamefont {Zhou}, \citenamefont {Xu}, \citenamefont {Su}, \citenamefont
  {Gao}, \citenamefont {Yan},\ and\ \citenamefont {Chen}}]{Meng:PRA:2018}%
  \BibitemOpen
  \bibfield  {author} {\bibinfo {author} {\bibfnamefont {H.-X.}\ \bibnamefont
  {Meng}}, \bibinfo {author} {\bibfnamefont {J.}~\bibnamefont {Zhou}}, \bibinfo
  {author} {\bibfnamefont {Z.-P.}\ \bibnamefont {Xu}}, \bibinfo {author}
  {\bibfnamefont {H.-Y.}\ \bibnamefont {Su}}, \bibinfo {author} {\bibfnamefont
  {T.}~\bibnamefont {Gao}}, \bibinfo {author} {\bibfnamefont {F.-L.}\
  \bibnamefont {Yan}}, \ and\ \bibinfo {author} {\bibfnamefont {J.-L.}\
  \bibnamefont {Chen}},\ }\href {\doibase 10.1103/PhysRevA.98.062103}
  {\bibfield  {journal} {\bibinfo  {journal} {Phys. Rev. A}\ }\textbf {\bibinfo
  {volume} {98}},\ \bibinfo {pages} {062103} (\bibinfo {year}
  {2018})}\BibitemShut {NoStop}%
\bibitem [{\citenamefont {Rabelo}\ \emph {et~al.}(2012)\citenamefont {Rabelo},
  \citenamefont {Law},\ and\ \citenamefont {Scarani}}]{Rabelo12}%
  \BibitemOpen
  \bibfield  {author} {\bibinfo {author} {\bibfnamefont {R.}~\bibnamefont
  {Rabelo}}, \bibinfo {author} {\bibfnamefont {Y.~Z.}\ \bibnamefont {Law}}, \
  and\ \bibinfo {author} {\bibfnamefont {V.}~\bibnamefont {Scarani}},\ }\href
  {\doibase 10.1103/PhysRevLett.109.180401} {\bibfield  {journal} {\bibinfo
  {journal} {Phys. Rev. Lett.}\ }\textbf {\bibinfo {volume} {109}},\ \bibinfo
  {pages} {180401} (\bibinfo {year} {2012})}\BibitemShut {NoStop}%
\bibitem [{\citenamefont {Chen}\ \emph {et~al.}(2023)\citenamefont {Chen},
  \citenamefont {Tabia}, \citenamefont {Jebarathinam}, \citenamefont {Mal},
  \citenamefont {Wu},\ and\ \citenamefont {Liang}}]{Chen23}%
  \BibitemOpen
  \bibfield  {author} {\bibinfo {author} {\bibfnamefont {K.-S.}\ \bibnamefont
  {Chen}}, \bibinfo {author} {\bibfnamefont {G.~N.~M.}\ \bibnamefont {Tabia}},
  \bibinfo {author} {\bibfnamefont {C.}~\bibnamefont {Jebarathinam}}, \bibinfo
  {author} {\bibfnamefont {S.}~\bibnamefont {Mal}}, \bibinfo {author}
  {\bibfnamefont {J.-Y.}\ \bibnamefont {Wu}}, \ and\ \bibinfo {author}
  {\bibfnamefont {Y.-C.}\ \bibnamefont {Liang}},\ }\href {\doibase 10.22331/q-2023-07-11-1054} {\bibfield  {journal} {\bibinfo
  {journal} {{Quantum}}\ }\textbf {\bibinfo {volume} {7}},\ \bibinfo {pages}
  {1054} (\bibinfo {year} {2023})}\BibitemShut {NoStop}%
\bibitem [{\citenamefont {Wootters}(1998)}]{Wootters:PRL:1998}%
  \BibitemOpen
  \bibfield  {author} {\bibinfo {author} {\bibfnamefont {W.~K.}\ \bibnamefont
  {Wootters}},\ }\href {\doibase 10.1103/PhysRevLett.80.2245} {\bibfield
  {journal} {\bibinfo  {journal} {Phys. Rev. Lett.}\ }\textbf {\bibinfo
  {volume} {80}},\ \bibinfo {pages} {2245} (\bibinfo {year}
  {1998})}\BibitemShut {NoStop}%
  \bibitem [{\citenamefont {Clauser}\ and\ \citenamefont
  {Horne}(1974)}]{Clauser74}%
  \BibitemOpen
  \bibfield  {author} {\bibinfo {author} {\bibfnamefont {J.~F.}\ \bibnamefont
  {Clauser}}\ and\ \bibinfo {author} {\bibfnamefont {M.~A.}\ \bibnamefont
  {Horne}},\ }\href {\doibase 10.1103/PhysRevD.10.526} {\bibfield  {journal}
  {\bibinfo  {journal} {Phys. Rev. D}\ }\textbf {\bibinfo {volume} {10}},\
  \bibinfo {pages} {526} (\bibinfo {year} {1974})}\BibitemShut {NoStop}%
\bibitem [{\citenamefont {Moroder}\ \emph {et~al.}(2013)\citenamefont
  {Moroder}, \citenamefont {Bancal}, \citenamefont {Liang}, \citenamefont
  {Hofmann},\ and\ \citenamefont {G\"uhne}}]{Moroder13}%
  \BibitemOpen
  \bibfield  {author} {\bibinfo {author} {\bibfnamefont {T.}~\bibnamefont
  {Moroder}}, \bibinfo {author} {\bibfnamefont {J.-D.}\ \bibnamefont {Bancal}},
  \bibinfo {author} {\bibfnamefont {Y.-C.}\ \bibnamefont {Liang}}, \bibinfo
  {author} {\bibfnamefont {M.}~\bibnamefont {Hofmann}}, \ and\ \bibinfo
  {author} {\bibfnamefont {O.}~\bibnamefont {G\"uhne}},\ }\href {\doibase 10.1103/PhysRevLett.111.030501} {\bibfield  {journal}
  {\bibinfo  {journal} {Phys. Rev. Lett.}\ }\textbf {\bibinfo {volume} {111}},\
  \bibinfo {pages} {030501} (\bibinfo {year} {2013})}\BibitemShut {NoStop}%
\bibitem [{Sup()}]{Supplemental}%
  \BibitemOpen
  \href@noop {} {}\bibinfo {howpublished} {See Supplemental Material at arXiv
  for the explicit quantum strategies leading to these success
  probabilities}\BibitemShut {NoStop}%
  \bibitem [{\citenamefont {Navascu\'es}\ \emph {et~al.}(2007)\citenamefont
  {Navascu\'es}, \citenamefont {Pironio},\ and\ \citenamefont
  {Ac\'{\i}n}}]{NPA}%
  \BibitemOpen
  \bibfield  {author} {\bibinfo {author} {\bibfnamefont {M.}~\bibnamefont
  {Navascu\'es}}, \bibinfo {author} {\bibfnamefont {S.}~\bibnamefont
  {Pironio}}, \ and\ \bibinfo {author} {\bibfnamefont {A.}~\bibnamefont
  {Ac\'{\i}n}},\ }\href {\doibase 10.1103/PhysRevLett.98.010401} {\bibfield
  {journal} {\bibinfo  {journal} {Phys. Rev. Lett.}\ }\textbf {\bibinfo
  {volume} {98}},\ \bibinfo {pages} {010401} (\bibinfo {year}
  {2007})}\BibitemShut {NoStop}%
\bibitem [{\citenamefont {Navascu{\'e}s}\ \emph {et~al.}(2008)\citenamefont
  {Navascu{\'e}s}, \citenamefont {Pironio},\ and\ \citenamefont
  {Ac{\'\i}n}}]{NPA2008}%
  \BibitemOpen
  \bibfield  {author} {\bibinfo {author} {\bibfnamefont {M.}~\bibnamefont
  {Navascu{\'e}s}}, \bibinfo {author} {\bibfnamefont {S.}~\bibnamefont
  {Pironio}}, \ and\ \bibinfo {author} {\bibfnamefont {A.}~\bibnamefont
  {Ac{\'\i}n}},\ }\href {http://stacks.iop.org/1367-2630/10/i=7/a=073013}
  {\bibfield  {journal} {\bibinfo  {journal} {New J. Phys.}\ }\textbf {\bibinfo
  {volume} {10}},\ \bibinfo {pages} {073013} (\bibinfo {year}
  {2008})}\BibitemShut {NoStop}%
\bibitem [{\citenamefont {P\"utz}\ \emph {et~al.}(2014)\citenamefont {P\"utz},
  \citenamefont {Rosset}, \citenamefont {Barnea}, \citenamefont {Liang},\ and\
  \citenamefont {Gisin}}]{PRB:PRL:2014}%
  \BibitemOpen
  \bibfield  {author} {\bibinfo {author} {\bibfnamefont {G.}~\bibnamefont
  {P\"utz}}, \bibinfo {author} {\bibfnamefont {D.}~\bibnamefont {Rosset}},
  \bibinfo {author} {\bibfnamefont {T.~J.}\ \bibnamefont {Barnea}}, \bibinfo
  {author} {\bibfnamefont {Y.-C.}\ \bibnamefont {Liang}}, \ and\ \bibinfo
  {author} {\bibfnamefont {N.}~\bibnamefont {Gisin}},\ }\href {\doibase 10.1103/PhysRevLett.113.190402} {\bibfield  {journal}
  {\bibinfo  {journal} {Phys. Rev. Lett.}\ }\textbf {\bibinfo {volume} {113}},\
  \bibinfo {pages} {190402} (\bibinfo {year} {2014})}\BibitemShut {NoStop}%
\bibitem [{\citenamefont {Goh}\ \emph {et~al.}(2018)\citenamefont {Goh},
  \citenamefont {Kaniewski}, \citenamefont {Wolfe}, \citenamefont {V\'ertesi},
  \citenamefont {Wu}, \citenamefont {Cai}, \citenamefont {Liang},\ and\
  \citenamefont {Scarani}}]{Goh2018}%
  \BibitemOpen
  \bibfield  {author} {\bibinfo {author} {\bibfnamefont {K.~T.}\ \bibnamefont
  {Goh}}, \bibinfo {author} {\bibfnamefont {J.}~\bibnamefont {Kaniewski}},
  \bibinfo {author} {\bibfnamefont {E.}~\bibnamefont {Wolfe}}, \bibinfo
  {author} {\bibfnamefont {T.}~\bibnamefont {V\'ertesi}}, \bibinfo {author}
  {\bibfnamefont {X.}~\bibnamefont {Wu}}, \bibinfo {author} {\bibfnamefont
  {Y.}~\bibnamefont {Cai}}, \bibinfo {author} {\bibfnamefont {Y.-C.}\
  \bibnamefont {Liang}}, \ and\ \bibinfo {author} {\bibfnamefont
  {V.}~\bibnamefont {Scarani}},\ }\href {\doibase 10.1103/PhysRevA.97.022104}
  {\bibfield  {journal} {\bibinfo  {journal} {Phys. Rev. A}\ }\textbf {\bibinfo
  {volume} {97}},\ \bibinfo {pages} {022104} (\bibinfo {year}
  {2018})}\BibitemShut {NoStop}%
\bibitem [{\citenamefont {Rai}\ \emph {et~al.}(2019)\citenamefont {Rai},
  \citenamefont {Duarte}, \citenamefont {Brito},\ and\ \citenamefont
  {Chaves}}]{Rai:PRA:2019}%
  \BibitemOpen
  \bibfield  {author} {\bibinfo {author} {\bibfnamefont {A.}~\bibnamefont
  {Rai}}, \bibinfo {author} {\bibfnamefont {C.}~\bibnamefont {Duarte}},
  \bibinfo {author} {\bibfnamefont {S.}~\bibnamefont {Brito}}, \ and\ \bibinfo
  {author} {\bibfnamefont {R.}~\bibnamefont {Chaves}},\ }\href {\doibase 10.1103/PhysRevA.99.032106} {\bibfield  {journal} {\bibinfo
  {journal} {Phys. Rev. A}\ }\textbf {\bibinfo {volume} {99}},\ \bibinfo
  {pages} {032106} (\bibinfo {year} {2019})}\BibitemShut {NoStop}%
\bibitem [{\citenamefont {Colbeck}\ and\ \citenamefont
  {Renner}(2012)}]{CR:NatPhys:2012}%
  \BibitemOpen
  \bibfield  {author} {\bibinfo {author} {\bibfnamefont {R.}~\bibnamefont
  {Colbeck}}\ and\ \bibinfo {author} {\bibfnamefont {R.}~\bibnamefont
  {Renner}},\ }\href {\doibase 10.1038/nphys2300} {\bibfield  {journal}
  {\bibinfo  {journal} {Nat. Phys.}\ }\textbf {\bibinfo {volume} {8}},\
  \bibinfo {pages} {450} (\bibinfo {year} {2012})}\BibitemShut {NoStop}%
\bibitem [{\citenamefont {Ramanathan}\ \emph {et~al.}(2018)\citenamefont
  {Ramanathan}, \citenamefont {Horodecki}, \citenamefont {Anwer}, \citenamefont
  {Pironio}, \citenamefont {Horodecki}, \citenamefont {Gr{\"u}nfeld},
  \citenamefont {Muhammad}, \citenamefont {Bourennane},\ and\ \citenamefont
  {Horodecki}}]{Ravi18}%
  \BibitemOpen
  \bibfield  {author} {\bibinfo {author} {\bibfnamefont {R.}~\bibnamefont
  {Ramanathan}}, \bibinfo {author} {\bibfnamefont {M.}~\bibnamefont
  {Horodecki}}, \bibinfo {author} {\bibfnamefont {H.}~\bibnamefont {Anwer}},
  \bibinfo {author} {\bibfnamefont {S.}~\bibnamefont {Pironio}}, \bibinfo
  {author} {\bibfnamefont {K.}~\bibnamefont {Horodecki}}, \bibinfo {author}
  {\bibfnamefont {M.}~\bibnamefont {Gr{\"u}nfeld}}, \bibinfo {author}
  {\bibfnamefont {S.}~\bibnamefont {Muhammad}}, \bibinfo {author}
  {\bibfnamefont {M.}~\bibnamefont {Bourennane}}, \ and\ \bibinfo {author}
  {\bibfnamefont {P.}~\bibnamefont {Horodecki}},\ }\href
  {https://doi.org/10.48550/arXiv.1810.11648} {\bibfield  {journal} {\bibinfo
  {journal} {arXiv:1810.11648}\ } (\bibinfo {year} {2018})}\BibitemShut
  {NoStop}%
\bibitem [{\citenamefont {Zhao}\ \emph {et~al.}(2023)\citenamefont {Zhao},
  \citenamefont {Ramanathan}, \citenamefont {Liu},\ and\ \citenamefont
  {Horodecki}}]{ZRL+23}%
  \BibitemOpen
  \bibfield  {author} {\bibinfo {author} {\bibfnamefont {S.}~\bibnamefont
  {Zhao}}, \bibinfo {author} {\bibfnamefont {R.}~\bibnamefont {Ramanathan}},
  \bibinfo {author} {\bibfnamefont {Y.}~\bibnamefont {Liu}}, \ and\ \bibinfo
  {author} {\bibfnamefont {P.}~\bibnamefont {Horodecki}},\ }\href {\doibase 10.22331/q-2023-09-15-1114} {\bibfield  {journal} {\bibinfo
  {journal} {{Quantum}}\ }\textbf {\bibinfo {volume} {7}},\ \bibinfo {pages}
  {1114} (\bibinfo {year} {2023})}\BibitemShut {NoStop}%
\bibitem [{\citenamefont {Rai}\ \emph {et~al.}(2022)\citenamefont {Rai},
  \citenamefont {Pivoluska}, \citenamefont {Sasmal}, \citenamefont {Banik},
  \citenamefont {Ghosh},\ and\ \citenamefont {Plesch}}]{Rai:PRA:2022}%
  \BibitemOpen
  \bibfield  {author} {\bibinfo {author} {\bibfnamefont {A.}~\bibnamefont
  {Rai}}, \bibinfo {author} {\bibfnamefont {M.}~\bibnamefont {Pivoluska}},
  \bibinfo {author} {\bibfnamefont {S.}~\bibnamefont {Sasmal}}, \bibinfo
  {author} {\bibfnamefont {M.}~\bibnamefont {Banik}}, \bibinfo {author}
  {\bibfnamefont {S.}~\bibnamefont {Ghosh}}, \ and\ \bibinfo {author}
  {\bibfnamefont {M.}~\bibnamefont {Plesch}},\ }\href {\doibase 10.1103/PhysRevA.105.052227} {\bibfield  {journal} {\bibinfo
  {journal} {Phys. Rev. A}\ }\textbf {\bibinfo {volume} {105}},\ \bibinfo
  {pages} {052227} (\bibinfo {year} {2022})}\BibitemShut {NoStop}%
\bibitem [{\citenamefont {Liu}\ \emph {et~al.}(2023)\citenamefont {Liu},
  \citenamefont {Chung},\ and\ \citenamefont {Ramanathan}}]{Liu:2309.06304}%
  \BibitemOpen
  \bibfield  {author} {\bibinfo {author} {\bibfnamefont {Y.}~\bibnamefont
  {Liu}}, \bibinfo {author} {\bibfnamefont {H.~Y.}\ \bibnamefont {Chung}}, \
  and\ \bibinfo {author} {\bibfnamefont {R.}~\bibnamefont {Ramanathan}},\
  }\href {https://arxiv.org/abs/2309.06304} {\bibfield  {journal} {\bibinfo
  {journal} {arXiv:2309.06304}\ } (\bibinfo {year} {2023})}\BibitemShut
  {NoStop}%
\bibitem [{\citenamefont {\v{S}upi\'c}\ and\ \citenamefont
  {Bowles}(2020)}]{Supic19}%
  \BibitemOpen
  \bibfield  {author} {\bibinfo {author} {\bibfnamefont {I.}~\bibnamefont
  {\v{S}upi\'c}}\ and\ \bibinfo {author} {\bibfnamefont {J.}~\bibnamefont
  {Bowles}},\ }\href {\doibase 10.22331/q-2020-09-30-337} {\bibfield  {journal}
  {\bibinfo  {journal} {{Quantum}}\ }\textbf {\bibinfo {volume} {4}},\ \bibinfo
  {pages} {337} (\bibinfo {year} {2020})}\BibitemShut {NoStop}%
\end{thebibliography}
\end{document}